\documentclass{article}
\pdfoutput=1

\usepackage[english]{babel}
\usepackage[paperwidth=225mm,      
paperheight=279mm,        
top=1.95cm,        
bottom=1.95cm,      
left=2.85cm,         
right=2.85cm,       
marginparwidth=1.75cm]{geometry}

\usepackage{amsmath,amssymb,amsthm}
\usepackage{bbm}
\usepackage{amscd}

\allowdisplaybreaks[4]

\usepackage{indentfirst}

\usepackage{authblk}

\usepackage{abstract}

\usepackage[bookmarksnumbered]{hyperref}
\usepackage{appendix}
\usepackage{natbib}
\usepackage{pdfpages}

\usepackage{tikz}
\usepackage{subcaption}
\usetikzlibrary{positioning, fit, shapes.misc,arrows.meta, matrix}
\usetikzlibrary{calc} 
\usepackage{tikz-cd}

\usepackage{algorithm}
\usepackage{algpseudocode}

\usepackage{booktabs} 
\usepackage{array}    
\usepackage{multirow} 

\usepackage{float} 
\usepackage{graphicx}

\numberwithin{equation}{section}
\numberwithin{figure}{section}

\usepackage{cases} 
\usepackage{empheq} 

\usetikzlibrary{matrix}

\usepackage{thmtools,thm-restate}
\usepackage{cleveref}

\newtheorem{Satz}{Theorem}[section]
\newtheorem{Prop}[Satz]{Proposition}
\newtheorem{Lem}[Satz]{Lemma}

\newtheorem{Cor}[Satz]{Corollary}

\theoremstyle{definition}

\newtheorem{Dfn}[Satz]{Definition}

\newtheorem{remark}[Satz]{Remark}
\newtheorem{example}[Satz]{Example}

\author{Jiaxing Weng\thanks{J. Weng: School of Economics, Ocean University of China, Qingdao, 266100, China; wengjiaxing@stu.ouc.edu.cn. I am deeply indebted to Tongyu Wang for his invaluable guidance in shaping the financial intuition behind this work. I also thank Jia Song, who brought this problem back to my attention. I am also grateful to Dunhuang City, where I spent a wonderful summer and first came to appreciate the flavor of the problem. }}

\title{Conservation of Short-term Flows: Signed Optimal Transport}

\begin{document}
	\maketitle
	\begin{abstract}
		\begin{sloppypar}
			This paper develops a theoretical framework for signed optimal transport. A global flatness measure induced by the continuum transport equation serves as the regularizer. We examine transport networks from a harmonic analysis perspective, prove the existence and uniqueness of the optimal coupling in the variational problem, and provide an algorithmic scheme that guarantees lossless information transport under bi-marginal constraints. To bridge theory with practice, we design an empirical analysis framework for the resulting optimal estimators, which is sufficiently general to accommodate both time-series and panel-structural analyses of the network, owing to the intrinsic invariance properties of locally compact abelian groups. 
		\end{sloppypar}
	\end{abstract}
	
	
\newpage
\section{Introduction}\label{sect: intro}
Optimal Transport (Henceforth OT) is well established for recovering the optimal coupling $\pi_\varepsilon^* \in \mathbb{R}^+\times \mathbb{R}^+$ from marginal distributions. Its classical formulation is built upon nonnegative measures $p,q>0$, which, after normalization, endow the space with a probability manifold structure. Such nonnegativity is natural when the transported objects are physical mass or probability. However, when the observables of interest are net flows, Proposition \ref{prop: unboundedness} demonstrates that OT problem lacks compactness if $\pi\in \mathbb{R}\times \mathbb{R}$. A natural idea is to move OT into complex domain, yet the entropy regularization function becomes singular there. 

The signed structure is not a theoretically isolated construct; rather, it arises naturally in various dynamic networks. For instance, in financial networks, short-term liquidity fluctuations, changes in interbank risk exposures, and portfolio rebalancing all manifest signed variations in balance sheets. Likewise,inventory adjustments in supply chains, species migration balances in ecological networks, and deviations from reference flows in transportation systems are all intrinsically signed. In these cases, the state spaces no longer correspond to probability measures, but rather to signed measures on product spaces.

This paper develops two frameworks. The \textit{theoretical framework} is designed for the analysis of signed OT, for which we prove the existence and uniqueness of the optimal coupling $\pi_\varepsilon^*$. In the decoupling algorithm and iterative scheme proposed to compute $\pi_\varepsilon^*$, we show that lossless information transport, the well-definedness of the update rule, and the convergence of the iterative process are equivalent. In addition, this paper effects a \textit{paradigm shift} in the understanding of network structures: rather than treating networks merely as matrix operators, we move beyond this conventional perspective to adopt an intrinsic group-theoretic structures. 

\begin{figure}[H]
	\centering
	\begin{subfigure}[b]{0.42\textwidth}
		\includegraphics[width=\textwidth]{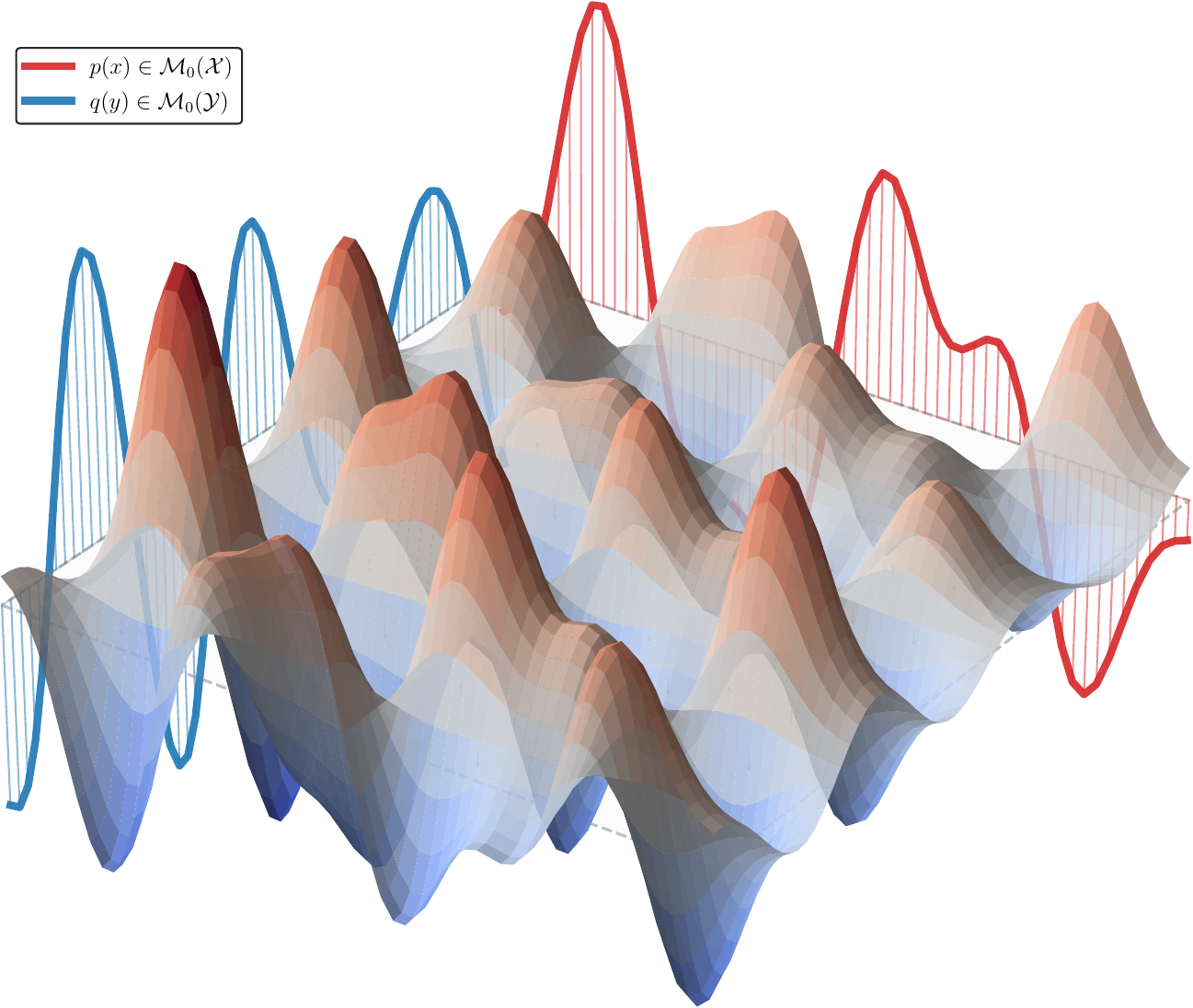}
		\caption{Normal Times}
	\end{subfigure}
	\hspace{0.5cm}
	\begin{subfigure}[b]{0.45\textwidth}
		\includegraphics[width=\textwidth]{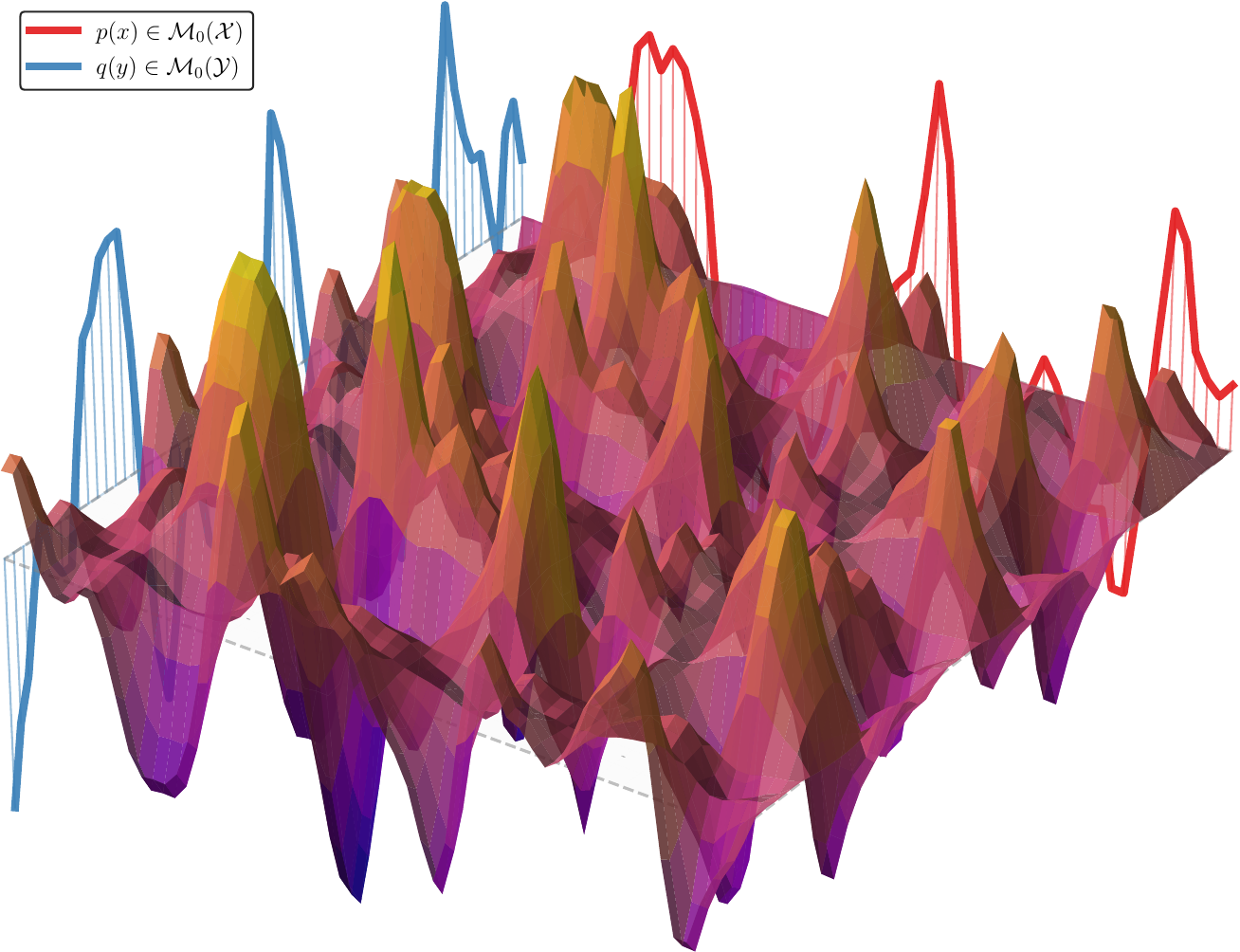}
		\caption{Financial Crises}
	\end{subfigure}
	\caption{Signed Optimal Transport}
	\label{fig: Signed Optimal Transport 3D}
\end{figure}

Once $\pi_\varepsilon^*$ has been estimated, we develop an \textit{empirical analysis framework} tailored to the structural properties of the underlying space. This framework enables us to: (1) characterize and precisely decompose the liquidity modes of the network; (2) detect arbitrary prior structures (e.g., core–periphery structures); (3) quantify both individual banks' risk exposures and the systemic risk of the banking system; and (4) track structural changes in the network over time—for instance, to identify whether new trading patterns emerge over the cycle from peacetime to the outbreak, propagation, and resolution of a financial crisis (see Figure \ref{fig: Signed Optimal Transport 3D}). This allows us to investigate how the mechanisms of interbank risk contagion evolve before and after financial crises, and to pinpoint exactly which propagation scales and which subgraph structures undergo liquidity flow patterns change. We elaborate on the theoretical framework in the main text, while the empirical analysis framework is presented in full in the Online Appendix.

Section 2 introduces the motivation and solution strategy for signed OT. There are various ways to restore compactness in the classical setting; but regularization demands a form of compactification that preserves data fidelity as much as possible. We investigate the generalizable feature underlying the maximum entropy principle, namely flatness, and, based on this insight, independently construct two curvature-based metrics. Since $\pi$ is signed, the notion of expectation ceases to apply, and with it the Fisher information matrix induced by the KL divergence is no longer available. Nevertheless, by taking a second-order expansion of the entropy regularization function and invoking the mass conservation constraint, we extract a local geometric structure $\|f\|_{L^2}^2$, which precisely characterizes the curvature information. However, this metric measures rigid pixel-wise variations: it registers only the total magnitude of perturbations, but fails to capture the dependence or correlation structure across different points. 

To extract further information, we reformulate the mass conservation constraint via the continuum transport equation for flow fields as introduced in \citet{chizat2018unbalanced}, which yields a formal representation of instantaneous density perturbations: 
$$
\partial_t\rho + \nabla\cdot(\rho v)=0.
$$
We set this as the initial perturbation $f$ and equip it with the norm $\|f\|_{\mathcal{H}^{-1}}^2$, thereby obtaining a global curvature measure $\int_\Omega|\nabla u|^2$. Spatial correlations are embedded into this measure naturally through the inverse of the Laplace operator. After considering the characteristics of systemic risk in financial networks, we argue that the curvature information encoded by $\|f\|_{\mathcal{H}^{-1}}^2$ is more suitable. However, since the underlying space on which it operates demands rather stringent conditions, we treat it here only as a heuristic guide. Nevertheless, the generator operator abstracted from its features will later serve as the regularization term in the variational problem.

Our goal is to design a regularization term $\mathcal H(\pi)$ on the underlying spaces $(\mathcal{X}, \mu)$ and $(\mathcal{Y}, \nu)$ such that the regularized signed OT \ref{eq: RegularizedProblem} admits a unique minimizer. At the same time, $\mathcal H(\pi)$ should mirror the behavior of Laplace operator in capturing global curvature information, so that the resulting optimal coupling $\pi_\varepsilon^*$ is as flat and sparse as possible.
\begin{equation}\label{eq: RegularizedProblem}
	\inf_{\pi \in \Pi_{\mathcal E_0}(p,q)} J_\varepsilon(\pi) := \int_{\mathcal{X}\times\mathcal{Y}} c(x,y)\,d\pi(x,y) + \varepsilon \mathcal H(\pi).
\end{equation}

We artificially assign bank space IDs so that $\mathcal X$ and $\mathcal Y$ are endowed with the structure of locally compact abelian groups, where bank IDs mapped to group elements via a mapping $\mathbb{N} \to \mathcal{X} \times \mathcal{Y}$. The detailed mapping procedure for the discrete cyclic group $\mathbb{Z}_2 \times \mathbb{Z}_2$ is provided in the Online Appendix. In fact, under the zero mass conservation setting, the structures of $\mathcal X$ and $\mathcal Y$ closely resemble those of compact groups, and many lemmas and propositions can be relaxed to finite measure spaces. Therefore, the group-theoretic setup is introduced primarily to facilitate the empirical framework and to streamline the proofs. Most importantly, of course, the algebraic properties of the operators, such as the spectral theorem, are intrinsic and do not entirely depend on the underlying space. Consequently, the functional form of $\pi_\varepsilon^*$, as well as the convergence and information preservation of the iterative algorithm, are primarily determined by the properties of the generator $L_{\mathcal X\times\mathcal Y}$.

Since a Riemannian gradient structure is generally unavailable on locally compact abelian groups, we induce a nonnegative, self-adjoint, and translation-invariant generator $L_{\mathcal X}$ on the space $\mathcal X$, so that it can play the role of the Laplace–Beltrami operator in measuring global curvature. In addition, the short-term flow conservation requires the generator to satisfy the Markov property, i.e., the spectral multiplier $m_{\mathcal X}(0) = 0$. We then further define the joint generator on the product space as $L_{\mathcal{X}\times\mathcal{Y}} = L_{\mathcal{X}} \otimes I + I \otimes L_{\mathcal{Y}}$, which inherits the Fourier characterization of the marginal generators (Lemma \ref{lem: UnitedGeneratorFourierTrans}).

We adopt the Dirichlet-type energy $\mathcal H(\pi) = \frac12 \|L_{\mathcal X\times\mathcal Y}^{1/2}\pi\|_{L^2}^2$ as the regularizer, where $\mathcal H(\pi)$ acts on the appropriate energy space $\mathcal E$:
$$
\mathcal E = D(L_{\mathcal X\times\mathcal Y}^{1/2}) = \left\{ \pi \in L^2(\mathcal{X}\times\mathcal{Y}) : \int m(\xi,\eta)|\hat\pi(\xi,\eta)|^2\,d\hat\mu d\hat\nu < \infty\right\}.
$$
This is a Hilbert space (Proposition \ref{prop: EnergySpaceIsHilbert}), which guarantees that $\mathcal H(\pi)$ exists and is the unique closed continuous quadratic extension, that is, $\mathcal H(\pi)$ is well-defined. If $\pi$ is $D(L_{\mathcal X\times\mathcal Y})$-regular, then
$$
\langle \pi, L_{\mathcal X\times\mathcal Y}\,\pi \rangle_{L^2}= \langle L_{\mathcal X\times\mathcal Y}^{1/2}\pi, L_{\mathcal X\times\mathcal Y}^{1/2}\pi \rangle_{L^2}= \|L_{\mathcal X\times\mathcal Y}^{1/2}\pi\|_{L^2}^2.
$$
We restrict the variational problem \ref{eq: RegularizedProblem} to $\mathcal E_0 = \mathcal M_0 \cap \mathcal E$, who collects signed density elements with finite spectral energy and zero total mass. The space $\mathcal E_0$ is again Hilbert, and $\mathcal H(\pi)$ is strictly convex on $\mathcal E_0$ (Lemma \ref{lem: N-EspaceHilbert}, Lemma \ref{lem: ConvexityOfDirletEnergy}). 
The set $\Pi_{\mathcal E_0}(p,q)$ of measures with marginal constraints is a closed convex set; hence, if a solution to \ref{eq: RegularizedProblem} exists, it is necessarily unique. Existence, however, requires a stronger priori control. We additionally require the generator to satisfy the spectral gap condition $m(\xi,\eta) \ge \lambda > 0$ for all $(\xi,\eta) \neq (0,0)$. This condition yields the coercivity of $L_{\mathcal X\times\mathcal Y}$, which not only guarantees the existence of a minimizer for \ref{eq: RegularizedProblem}, but also indirectly ensures lossless information transport in the iterative algorithm (Proposition \ref{prop: optimalEstimationExistUniq}).

We next derive, by means of the calculus of variations, the local weak Euler–Lagrange equation (henceforth, the E–L equation), which is defined only on $\ker A$ of the marginal operator $A$. Our ultimate goal, however, is to obtain a globally defined iterative algorithm that allows us to estimate the optimal coupling from arbitrary initial points. To extend the weak E–L equation from $\ker A$ to all of $\mathcal E_0$, we prove that when $\mathcal X$ and $\mathcal Y$ are compact groups, the image $\operatorname{Im} A$ is closed, and we identify its precise range as $\operatorname{Im} A = \mathcal E_{\mathcal X}^0 \times \mathcal E_{\mathcal Y}^0$, along with the domain of the adjoint operator $A^*$, namely $(\mathcal E_{\mathcal X}^0 \times \mathcal E_{\mathcal Y}^0)^*$ (Lemma \ref{lem: Closed Image}). By the duality isomorphism on product spaces, we obtain the duality principle $A^*(\phi^*, \psi^*) = \phi^* + \psi^*$ (Lemma \ref{lem: Adjoint Operator}). Applying the closed range theorem, we then express the first variation $DJ_\varepsilon(\pi^*)$ of \ref{eq: RegularizedProblem} as the functional $A^*$, thereby yielding the global weak E–L equation (Proposition \ref{prop: Weak-EL-Equation}). The proof of the global strong E–L equation relies on the $D(L_{\mathcal X\times\mathcal Y})$-regularity of $\pi$ and a density lemma. The latter shows that $\mathcal E_0$ is dense in $\mathcal M_0$—that is, every $L^2$ state with zero total mass but possibly infinite energy can be approximated by states of zero total mass and finite energy (Lemma \ref{lem: DenseLemma}). Finally, leveraging the geometric structure of $\mathcal M_0$, we extend the equation to the $L^2$ space (Proposition \ref{prop: Strong-EL-Equation}).

As with most variational problems, the optimal coupling triple $(\pi^*, \phi^*, \psi^*)$ satisfies not only the saddle-point condition (Proposition \ref{prop: Saddle Point}),
$$
\mathcal L(\pi^*, \phi, \psi) \le \mathcal L(\pi^*, \phi^*, \psi^*) \le \mathcal L(\pi, \phi^*, \psi^*),\quad \forall \pi \in \mathcal E_0,\; \phi, \psi.
$$
but also strong duality (Proposition \ref{prop: StrongDuality}),
$$
\inf_{\pi} \sup_{\phi,\psi} \mathcal L(\pi, \phi, \psi)=\sup_{\phi,\psi} \inf_{\pi} \mathcal L(\pi, \phi, \psi)=\mathcal L(\pi^*, \phi^*, \psi^*).
$$
and this does not require any additional regularity of $\pi$. In fact, this shows that the variational problem \ref{eq: RegularizedProblem} intrinsically embeds a zero-sum game structure on a continuous strategy space. 

Finally, we reformulate the original signed transport problem as the operator equation involving the spectral Green's operator (Equation \ref{eq: DualOperaotr}):
\begin{equation}\label{eq: DualOperaotr}
	A G A^* z = \varepsilon b + A G c.
\end{equation}
The zero-mass conservation ensures that the Green's operator $G := L_{\mathcal X\times\mathcal Y}^{-1}$ is well-defined. By introducing the block operator $K := A G A^*$ and the constant $d := \varepsilon b + A G c$, Equation \ref{eq: DualOperaotr} is further transformed into a linear operator equation for the dual variables $z = (\phi, \psi)$:
\begin{equation}\label{eq: LinearDualOperaotr}
	K z = d.
\end{equation}
Proposition \ref{prop: RichardsonIteration} provides the convergence of the Richardson iteration, namely that the decoupling equation $z^{n+1} = (I - \alpha^* K)z^n + \alpha^* d$ converges to the unique $z^*$. However, this baseline model does not yield further information on the state of information transport. Exploiting the bi-marginal structure of signed OT \ref{eq: RegularizedProblem}, we further reduce the original transport plan and construct the error propagation equation $e_\psi^{n+1} = -K_{\mathcal Y\mathcal Y}^{-1} K_{\mathcal Y\mathcal X} \left( -K_{\mathcal X\mathcal X}^{-1} K_{\mathcal X\mathcal Y} e_\psi^n \right)$. This reduction reveals the intrinsic block structure of problem \ref{eq: RegularizedProblem}: $e_\psi^{n+1} = M e_\psi^n$, with convergence governed by the spectral radius condition $\rho(M) < 1$. For the recovery of the optimal coupling from marginal observations, however, the update rule must be well-defined, which requires $K_{\mathcal X\mathcal X} = P_\mathcal X G P_\mathcal X^*$ to be invertible. In addition, we must ensure that information is transported losslessly under the lifting operators $P_\mathcal X^*$ and $P_\mathcal Y^*$, that is, whether $\|P_\mathcal X^*\phi\| \ge C\|\phi\|$ holds. Lemma \ref{lem: CoercivityToInvertible}, Lemma \ref{lem: CoecivityOfLifitingOperator}, and Proposition \ref{prop: legitimacy ensures convergence} show that these conditions are equivalent and are automatically satisfied within the theoretical framework of this paper.

Figure \ref{fig: Framework} presents the overall framework of this paper.
\begin{figure}[H]
	\centering
	\includegraphics[width = 0.8\textwidth]{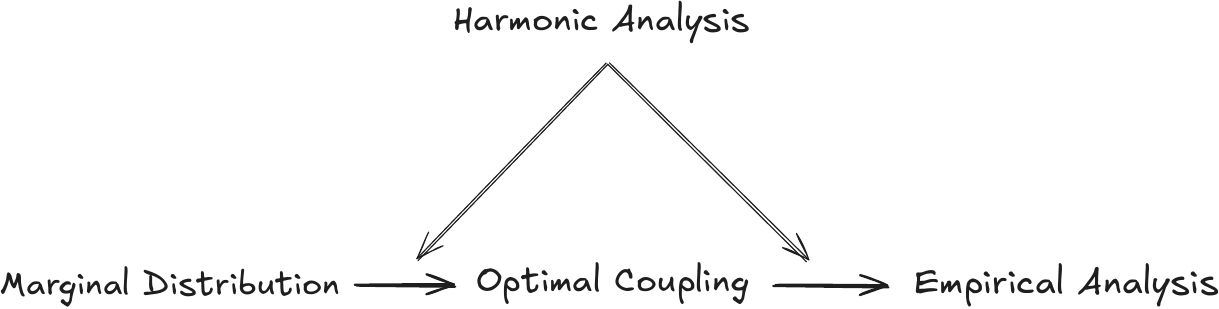}
	\caption{Framework}
	\label{fig: Framework}
\end{figure}

On the empirical side, our experimental paradigm differs markedly from conventional network analyses, which rely primarily on eigenvector-based (or singular-vector-based) centrality measures. The distinction is not merely a matter of a richer statistical toolkit or more refined traceability techniques. In our framework, banks and their funding positions are embedded in a state space endowed with a fixed group structure, for which the entire framework admits an invariant Fourier basis (i.e., characters). If network structures were to be treated merely as matrix operators, then even with the same set of banks across periods, their leading singular vectors $u_1$ would vary with the matrices $W_0, W_1$ as they evolve; any time-series analysis lacking a fixed benchmark reference would be inherently distorted. Moreover, the information that singular vectors can capture is rather limited. While singular values may characterize centrality, they can hardly capture phenomena such as long-term liquidity migration or regional risks (e.g., local shocks). Our concern is not simply whether the financial network structure has changed, but also to identify and track at what scale the structural transitions occur, and in what mode of risk propagation these transitions take place. 

\subsection{Literature Review}
This paper is largely motivated by the work of \citet{wang2025biggest}, who developed the Directed MaxEnt Method to estimate aggregate structural changes in banking networks during economic downturns. We abstract the underlying signed transport problem from their work, extend the algorithmic structure inherent in their estimation procedure, and provide a rigorous theoretical justification for it. 

The signed transport problem has been discussed in various contexts, for instance, by \citet{piccoli2014generalized} and \citet{mainini2012description}, who interpret signed measures as the setting in economic terms, where, for example, negative liabilities can be viewed as assets. The Wasserstein distance they adopt takes the form $\mathbb{W}_p(\mu, \nu) := W_p(\mu^+ + \nu^-,\ \nu^+ + \mu^-)$. \citet{mainini2012description} extend the optimal transport framework to the space of real-valued measures, with the primary motivation arising from the study of non-positive solutions to certain evolution PDEs. As for the choice of regularizers, it is by no means unique. \citet{blondel2018smooth} exploit the squared 2-norm to obtain sparse optimal plans; although their discussion remains within the probability simplex, their approach could in fact also play a role in the variational problem \ref{eq: RegularizedProblem}. The global flatness measure we derive relies on the continuum transport equation proposed by \citet{chizat2018unbalanced}. In their work, \citet{chizat2018unbalanced} develop a general theoretical framework for unbalanced optimal transport on nonnegative Radon measures, unifying the dynamic (fluid-dynamic) and static (Kantorovich) formulations, and establish their equivalence under suitable regularity conditions. For broader background and further intuition on optimal transport, we refer to \citet{villani2009optimal} and \citet{santambrogio2015optimal}.

\section{Entropy-Inspired Quantification of Flatness}\label{sect: entropy}
In this chapter, we introduce the motivation for studying the signed optimal transport problem. We construct curvature-based metrics to restore the missing compactness, while preserving as much as possible the flatness of the optimal coupling.
\subsection{Motivation}
\subsubsection{Classic Optimal Transport}
In the classical OT framework, we set $\pi\in X \times Y\subset \mathbb{R}\times \mathbb{R}$, this means that the distributions can take arbitrary real values, not merely positive ones. Let $\mathcal G$ denote the space of signed measures. Given marginal signed measures $\mu, \nu \in \mathbb R$ satisfying the mass conservation condition $\mu(X) = \nu(Y)$:
$$
\mathcal{G}=\inf_{\pi \in \mathcal{M}(X \times Y)} \left\{ \int c \, d\pi \;\Big|\; P_X\#\pi = \mu,\; P_Y\#\pi = \nu \right\},
$$
the feasible set is:
$$
\Pi(\mu,\nu) = \left\{ \pi \in X \times Y : P_X\#\pi = \mu,\; P_Y\#\pi = \nu \right\}.
$$
In the absence of nonnegativity constraints, the feasible set becomes an affine linear manifold. The solution takes the form $\pi = \mu \otimes \nu + q$, where $q$ is a signed measure with zero marginals. The optimization problem then becomes:
$$
\langle c, \pi \rangle = \langle c, \mu \otimes \nu \rangle + \langle c, q \rangle 
$$
Note that the first term is constant, so it suffices to optimize the second term. The zero-marginal measure $q$ satisfies the following equation:
$$
\int_y q_{xy} = 0, \,\,\, \int_x q_{xy} = 0, \,\,\, \forall x\in X,y\in Y.
$$
This problem is clearly ill-posed (Proposition \ref{prop: unboundedness}); that is, the classical OT formulation lacks compactness when faced with signed marginal measures. 
\begin{Prop}[Unboundedness Below]\label{prop: unboundedness}
	When $q \neq 0$, $\mathcal G = \varnothing$; when $\langle c, q \rangle = 0$, $|\mathcal G| = +\infty$.
\end{Prop}
\begin{proof}
	See Section \ref{sect: proof of unboundedness}.
\end{proof}

A natural idea is to add an entropy regularization term. However, in the classical entropy regularization approach, the entropy function $H(\rho)=-\int_{\Pi(\mu,\nu)} \rho\log\rho \,d\pi$ requires that the signed measure $\pi$ does not take values on the negative real half-axis $\mathbb{R}^-$, which prevents the above problem from being fully solved over the real domain. Nevertheless, the extension of $H(\rho)$ to the complex domain is well-defined. For example (with $\arg z = \pi$):
$$
\ln(-6) = \ln 6 + \ln(-1)  = \ln 6 + \pi i.
$$

\subsubsection{Complex Entropy} 
\begin{Lem} \label{lem: complexEntropy}
	For $\rho\in\mathbb C\setminus\{0\}$, we have:
	\begin{enumerate}
		\item[(1)] $H(\rho)=-\int_{\Pi(\mu,\nu)} \rho\log|\rho| \,d\pi -\pi i\int_{\Pi^-(\mu,\nu)} \rho\,d\pi$
		\item[(2)] $H(\rho)=-\int_{\Pi(\mu,\nu)} |\rho|\log|\rho| \,d\pi -2\int_{\Pi^-(\mu,\nu)} \rho\log|\rho| \,d\pi -\pi i\int_{\Pi^-(\mu,\nu)} \rho\,d\pi$.
	\end{enumerate}
\end{Lem}
It is easy to check that the real part of the complex entropy $\mathcal{H}_{\text{abs}}(\rho) = -\int \rho \log|\rho| \, d\pi$ is singular at $\rho = 0$. Moreover, consistent with the situation described in Proposition \ref{prop: unboundedness}, the entropy-regularized problem with this function suffers from unboundedness: the objective function is driven downward along the optimization direction due to discontinuous jumps in the optimal coupling, so that $\mathcal G = \varnothing$ and the optimization problem admits no solution. Note that this ill-posedness arises from the particular choice of the complex entropy function $H(\rho)$ itself, rather than from signed measures themselves. The second term of $H(\rho)$ arises entirely from the contribution of negative-density regions. When $-2\int_{\Pi^-(\mu,\nu)} \rho\log|\rho| \,d\pi=0$, the entropy regularization reduces to the classical Shannon entropy. In this case, $H_{\text{abs}}(\rho) = H_{\text{abs}}(-\rho)$, meaning that the energy is invariant under a global sign flip—i.e., positive and negative densities are energetically equivalent. Nevertheless, local singularities and jumps in the coupling persist, and there remain infinitely many feasible solutions matching the marginals, with the objective function unbounded below. Interestingly, if the entropy takes the form $H = |\pi| \log \rho$, then the original transport problem admits a projected iterative solution; see \citet{junius2003solution}. The above discussion of complex entropy regularization also suggests that eliminating the singularity is of crucial importance.

\subsubsection{Flatness}
A regularization form compatible with the signed transport problem is bound to exist. For the missing compactness in the original problem, one may simply add a coercive function to restore it. However, not all forms of compactness are equally desirable. The real question is: what kind of compactness best preserves the data fidelity? This calls for a more objective criterion. We therefore return to the positive-probability setting and ask what features are generalizable. 

We identify the flat structure underlying the maximum entropy principle, as revealed by both the Sinkhorn algorithm and the complex entropy results, and take it as our starting point. Recall the classical entropy regularization. When the regularization is strong, the optimal coupling approaches the flattest possible distribution, namely the independent coupling of the marginals, which is a flat linear structure. As the regularization weakens, the flatness of the optimal coupling gradually diminishes, and the solution degenerates to that of the unregularized problem:
$$
\lim_{\varepsilon\rightarrow \infty}P_{ij} = \lim_{\varepsilon\rightarrow \infty}a_i e^{-C_{ij}/\epsilon} b_j \to \mu_i \nu_j 
$$
Moreover, the real part of $H(\rho)$ naturally encodes a preference for flatness: $-\int_{\Pi(\mu,\nu)} |\rho|\log|\rho| \,d\pi$. 

\subsection{Local Geometry of Entropy}
Since the notion of expectation is no longer valid in this setting, we cannot proceed as in the standard derivation of the Fisher information matrix. For the case $\int_\Omega \rho \, dx = \gamma$ with $\rho \in L^2$, we perform a Taylor expansion of $H(\rho)$ and obtain:
$$
\begin{aligned}H(\rho_\varepsilon) &= -\int_\Omega (\rho_0+\varepsilon f)\log(\rho_0+\varepsilon f)\,dx \\ &= -\int_\Omega \left[
	\rho_0\log\rho_0
	+
	\varepsilon(1+\log\rho_0)f
	+
	\frac{\varepsilon^2}{2}\frac{f^2}{\rho_0}
	+
	O(\varepsilon^3)
	\right] dx. \end{aligned}
$$
The original system is subject to the constraint $\int_\Omega p_0 \, dx = \gamma$. The mass conservation constraint (for a closed system) requires that $\int_\Omega \rho_0 + \varepsilon f \,dx=\gamma$, which implies:
$$
\int_\Omega f  \,dx=0.
$$
We set the reference density to $\rho_0 \equiv \frac{\gamma}{|\Omega|}$, which corresponds to a uniform distribution of mass over the entire domain $\Omega$. By a scaling transformation, we normalize it to $\rho_0 \equiv 1$. This allows us to reformulate $H(\rho_\varepsilon)$ as a measure of the local geometric (curvature) information of the Shannon entropy function: $H(\rho_\varepsilon) =  -\frac{\varepsilon^2}{2}  \int_\Omega f^2=-\frac{\varepsilon^2}{2}\|f\|_{L^2}^2$, which exhibits properties analogous to those of the Fisher information metric or the K–L divergence.

\subsection{Continuum Transport Equation}
To obtain further analytical properties, we reformulate the mass conservation constraint via the continuum transport equation for flow fields, which additionally requires $f$ to satisfy:
$$\partial_t\rho + \nabla\cdot(\rho v)=0 \;\Longrightarrow\; f = -\nabla\cdot(\rho_0 v).$$
From the small perturbation $\rho(x,t) = \rho_0 + \varepsilon f(x,t)$, we obtain $\partial_t (\rho_0 + \varepsilon f) + \nabla \cdot [(\rho_0 + \varepsilon f) v] = 0$. Since the reference density satisfies $\partial_t \rho_0 = 0$, we have $\varepsilon \partial_t f + \nabla \cdot (\rho_0 v) + \varepsilon \nabla \cdot (f v) = 0$. By the arbitrariness of the velocity field, we set $v = \nabla u = O(\varepsilon)$, and neglect the higher-order term $\varepsilon \nabla \cdot (f v)$, yielding $\partial_t f  = - \frac{1}{\varepsilon}\nabla \cdot (\rho_0 v)$. Setting $\tau = t / \varepsilon$ yields $\partial_t=\frac1\varepsilon\partial_\tau$, and the equation becomes:
$$
\partial_\tau f
=
-\nabla\cdot(\rho_0v)
$$
It is clear that $-\nabla \cdot (\rho_0 v)$ is also a form of perturbation. We accordingly define $g$ as the instantaneous density perturbation induced by the velocity field:
$$
g = -\nabla \cdot (\rho_0 v)
$$
It is worth noting that the above notation is introduced solely to avoid conceptual confusion. In fact, one may define $f = -\nabla \cdot (\rho_0 v)$ from the very beginning, with $u \in \mathcal{H}^1_0(\Omega)$, so that:
$$
f = -\nabla \cdot (\nabla u) = -\Delta u
$$
In other words, $f \in \text{Im}(-\Delta) = \{ -\Delta u : u \in \mathcal{H}^1_0(\Omega) \}$. Let the dual space of $\mathcal{H}^{1}_0(\Omega)$ be denoted by $(\mathcal{H}^1_0)^* = \mathcal{H}^{-1}_0(\Omega)$. Then $-\Delta: \mathcal{H}^1_0(\Omega) \to \mathcal{H}^{-1}_0(\Omega)$. This implies that $f \in \mathcal{H}^{-1}_0(\Omega)$. We endow $\mathcal{H}^{-1}_0(\Omega)$ with the norm $\|f\|_{\mathcal{H}^{-1}}^2 = \langle f, (-\Delta)^{-1} f \rangle$, which yields:
$$
\begin{aligned}
	\|f\|_{\mathcal{H}^{-1}_0}^2 
	= \langle f, (-\Delta)^{-1} f \rangle &= \langle -\Delta u, (-\Delta)^{-1} (-\Delta u) \rangle \\
	&= \langle -\Delta u, u \rangle  \\
	&= \int_\Omega |\nabla u|^2 \, dx 
\end{aligned}
$$

In information geometry, the K-L divergence (or, equivalently, the local Fisher information metric) measures the local curvature of the space of probability distributions. This corresponds to the norm $\|f\|_{L^2}^2$, which quantifies pixel-wise variations of the density function. However, for economic problems, particularly those concerning systemic risk and trading stability, the global flatness measure offers greater economic insight than local curvature. While the $L^2$-norm focuses solely on the magnitude of local perturbations of the density, the norm $\|f\|_{\mathcal{H}^{-1}}^2$ measures, from a global structural perspective, the cost of spatial movement required to generate such perturbations. It is more flexible than the $L^2$-norm and inherently embeds spatial correlations through the inverse of the Laplace operator. Flatness, in turn, serves as a measure of global risk exposure. For these reasons, we adopt the latter as the regularization term in the variational problem, in order to capture the deep frictions inherent in liquidity transport and matching directions.

Moreover, the choice of the $\|f\|_{\mathcal{H}^{-1}}^2$ implies that our analytical framework can naturally map bilateral trading positions into distinct risk modes. These modes are not merely a coarse decomposition; rather, they serve to precisely identify and trace the propagation paths of liquidity, that is, the transport structure of risk. For trading paths that account for a significant proportion of total positions, the flatness measure $\|f\|_{\mathcal{H}^{-1}}^2$ imposes a penalty. This mechanism compels banks to engage in an endogenous trade-off between transaction costs and risk exposure, thereby steering industry’s liquidity profile toward a configuration that is as flat as possible. This configuration, in turn, constitutes the micro-foundation for minimal systemic risk. In contrast, $\|f\|_{L^2}^2$ treats the underlying space as rigid, with no distance cost between points; its Fourier representation is given by $\|f\|_{L^2}^2 = \int |\hat f(\xi)|^2 d\xi$, and therefore it treats all frequencies equally.

The essence of $\|f\|_{\mathcal{H}^{-1}}^2$ lies in its Fourier representation:
$$
\int|\nabla u|^2=\int|\xi|^2|\widehat u|^2,
$$
where $|\xi|$ denotes frequency and $\widehat u$ the Fourier modes. The $\mathcal{H}^{-1}$-norm penalizes high-frequency fluctuations; consequently, the flatter the function (i.e., the more it is dominated by low frequencies), the smaller $\int_\Omega |\nabla u|^2$. In this case, the $\mathcal{H}^{-1}$-norm is also small, which is consistent with the maximum entropy principle, the latter also favors the most uniform and unbiased distribution given the constraints. Figure \ref{fig: Comparison of flatness} provides a visual comparison of this flatness difference: the left panel exhibits a highly folded structure with dense spikes, corresponding to a large $\mathcal{H}^{-1}$-norm; the right panel, by contrast, shows a smoother structure with a lower level of global risk exposure, whose risk transmission paths are more dispersed and more resilient. Moreover, $\int_\Omega |\nabla u|^2$ can be interpreted as a weighted sum of the energies of the Fourier modes, which serves as the most basic idea underlying the subsequent modal decomposition.
\begin{figure}[H]
	\centering
	\begin{subfigure}[b]{0.44\textwidth}
		\includegraphics[width=\textwidth]{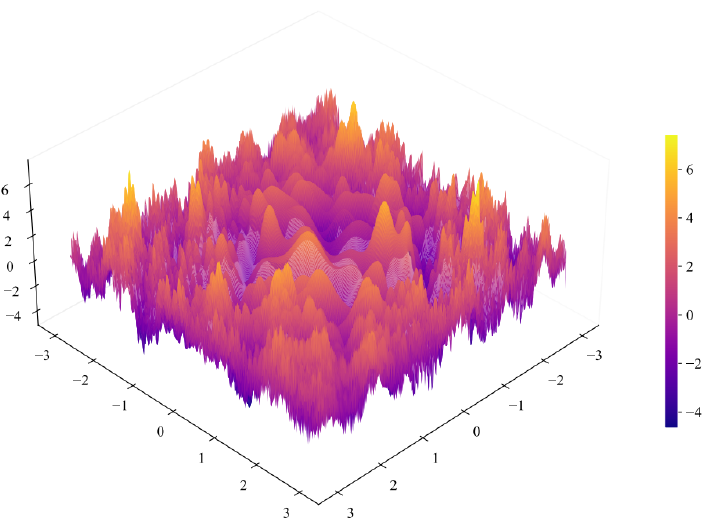}
		\caption{Highly Folded}
	\end{subfigure}
	\hspace{0.5cm}
	\begin{subfigure}[b]{0.45\textwidth}
		\includegraphics[width=\textwidth]{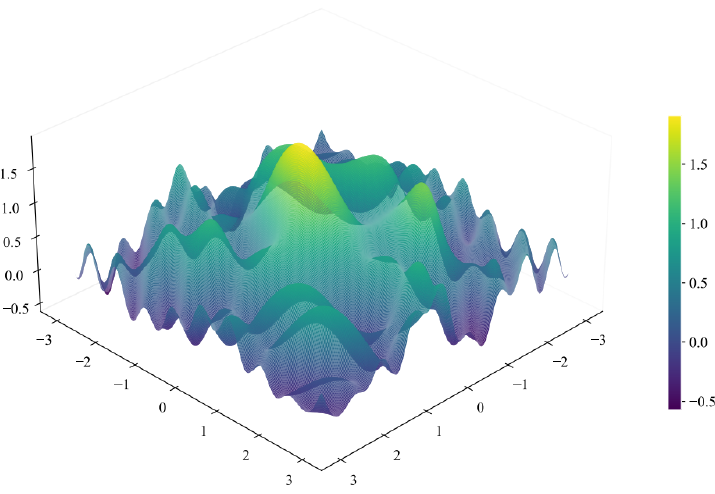}
		\caption{Smooth}
	\end{subfigure}
	\caption{Comparison of Flatness}
	\label{fig: Comparison of flatness}
\end{figure}

\section{Signed Optimal Transport}\label{sect: signedOT}
This section models the signed transport problem from the perspective of harmonic analysis\footnote{\citep{folland2016course} provides an excellent introduction.} and abstracts the previously introduced curvature measure as the Laplace operator, which then serves as the regularizer. 
\subsection{State Space}
Let the marginal spaces $(\mathcal{X}, \mu)$ and $(\mathcal{Y}, \nu)$ be locally compact abelian groups equipped with Haar measures $\mu$ and $\nu$, respectively, and let their Pontryagin dual groups be denoted by $\widehat{\mathcal{X}}$ and $\widehat{\mathcal{Y}}$. We begin by defining the space of signed measures $\mathcal{M}_0(\mathcal{X})$.

\begin{Dfn}[Zero-Mass State Space]
	$$
	\mathcal{M}_0(\mathcal{X}) = \left\{ p \in L^2(\mathcal{X}, \mu) : \int_{\mathcal{X}} p(x)\,d\mu(x) = 0 \right\}.
	$$
	i.e., $\mathcal{M}_0(\mathcal{X}) \neq \emptyset$ denotes the space of all signed densities with zero total mass.
\end{Dfn}

The space $\mathcal{M}_0(\mathcal{X})$ is a zero-mean hyperplane of codimension $1$ in $L^2(\mathcal{X}) = \int_{\widehat{\mathcal{X}}}^{\oplus} \mathbb{C} \cdot \chi_\xi \, d\hat{\mu}(\xi)$. Its definition retrospectively singles out the admissible locally compact group objects, such as compact groups. The space $\mathcal{M}_0(\mathcal{Y})$ is defined analogously. The total mass conservation condition $\int p = \int q = 0$ represents short-term liquidity conservation (total assets $=$ total liabilities).

\begin{remark}\label{remark: geometricIntuition}
	Geometrically, $L^2(\mathcal{X}) = \bigoplus\limits_{\xi \in \widehat{\mathcal{X}}} \mathbb{C} \cdot \chi_\xi = \text{span}\{1\} \oplus \mathcal{M}_0(\mathcal{X})$, and $\text{span}\{1\}$ correspond to the zero-frequency component $\xi=0$ in the Fourier spectrum. The space $\mathcal{M}_0(\mathcal{X})$, as its orthogonal complement, has vanishing spectrum at the origin (due to the trivial character $\chi_0(x) = 1$), i.e., $\mathcal{M}_0(\mathcal{X})$ retains all oscillatory modes at nonzero frequencies ($\xi\neq0$). 
	$$
	\hat{p}(0) = \int_{\mathcal{X}} p(x) \overline{\chi_0(x)}\,d\mu(x) = \int_{\mathcal{X}} p(x)\,d\mu(x) = 0.
	$$
	If $\text{span}\{1\}\subset L^2(\mathcal{X}\times\mathcal{Y})$ is required, one further needs the Haar measures to be finite, i.e., $\mathcal{X}$ and $\mathcal{Y}$ are compact.
	Moreover, for any $f\in L^2(\mathcal{X}\times\mathcal{Y})$, one may uniquely decompose it along the constant and zero-mean directions as $f=f_{\parallel}+f_{\perp}$, with $f_{\parallel}=a\mathbf{1}$ subject to $f-f_{\parallel}\in\mathcal M_0$, i.e., satisfying $\int_{\mathcal{X}\times\mathcal{Y}}(f-a)d\lambda=0$. Solving yields $a=\frac1{\lambda(\mathcal{X}\times\mathcal{Y})}\int_{\mathcal{X}\times\mathcal{Y}}f(z)d\lambda(x,y)$, which is the average of $f$ over $\mathcal{X}\times\mathcal{Y}$.
\end{remark}

\begin{remark}
	Plancherel theorem states that the Fourier transform preserves the $L^2$ inner product, i.e., $\mathcal{F}: L^2(\mathcal{X}) \to L^2(\widehat{\mathcal{X}})$ is an isometric isomorphism. Therefore, the Fourier-domain representation of $\mathcal{M}_0(\mathcal{X})$ is given by:
	$$
	\mathcal{M}_0(\mathcal{X}) \simeq \left\{ \hat{p} \in L^2(\widehat{\mathcal{X}}, \hat\mu) : \hat{p}(0) = 0 \right\}.
	$$
\end{remark}

\begin{remark}\label{remark: Invertible}
	The term $\hat{p}(0) = 0$ is used to remove the singularity of the equation at zero frequency. In the Fourier domain, the equation $Lp = f$ takes the form $\widehat{L p}(\xi) =m(\xi)\hat{p}(\xi) = \hat{f}(\xi)$ (with spectral multiplier $m(\xi) \ge 0$). The well-definedness of its inverse operator $\widehat{L^{-1} f}(\xi) = \frac{\hat{f}(\xi)}{m(\xi)}=\hat{p}(\xi)$ at $\xi = 0$ (where $m(\xi) \neq 0$ for $\xi \neq 0$) is guaranteed by $\hat{p}(0) = 0$, i.e., $\widehat{L^{-1} f}(0)=0$.
\end{remark}

\begin{Dfn}[Joint State Space and Transport Plan]
	Define the space of joint signed densities:
	$$
	\mathcal{M}_0(\mathcal{X}\times\mathcal{Y}) = \left\{ \pi \in L^2(\mathcal{X}\times\mathcal{Y}, \mu\otimes\nu) : \int_{\mathcal{X}\times\mathcal{Y}} \pi(x,y)\,d\mu(x)d\nu(y) = 0 \right\}.
	$$
	For given marginal densities $p \in \mathcal{M}_0(\mathcal{X})$ and $q \in \mathcal{M}_0(\mathcal{Y})$, define the set of transport plans (marginal constraints):
	$$
	\Pi_0(p,q) = \left\{ \pi \in \mathcal{M}_0(\mathcal{X}\times\mathcal{Y}) : P_{\mathcal X}\pi = p,\ P_{\mathcal Y}\pi = q \right\},
	$$
	where the projection operators $P_{\mathcal X}: L^2(\mathcal{X}\times\mathcal{Y}) \to L^2(\mathcal{X})$ and $P_{\mathcal Y}: L^2(\mathcal{X}\times\mathcal{Y}) \to L^2(\mathcal{Y})$ are defined by
	$$
	\begin{aligned}(P_{\mathcal X}\pi)(x) &= \int_{\mathcal{Y}} \pi(x,y)\,d\nu(y),\\
		(P_{\mathcal Y}\pi)(y) &= \int_{\mathcal{X}} \pi(x,y)\,d\mu(x). \end{aligned}
	$$
\end{Dfn}

\subsection{Fourier Transform}
On a locally compact abelian group $\mathcal X$, for each $\xi \in \widehat{\mathcal{X}}$, we define the character (irreducible unitary representation):
$$
\chi_\xi(x) = \langle \xi, x \rangle,
$$
where $\langle \cdot, \cdot \rangle: \widehat{\mathcal{X}} \times \mathcal{X} \to \mathbb{T}$ denotes the Pontryagin pairing. In the special case of Lie groups such as $\mathcal X = \mathbb{R}^d$, this pairing reduces to $\langle \xi, x \rangle = e^{2\pi i \xi \cdot x}$.

For $p \in \mathcal{M}_0(\mathcal{X}) \subset L^2(\mathcal{X})$, the Fourier transform and its inverse are given by:
$$
\begin{aligned} \hat{p}(\xi) &= \int_{\mathcal{X}} p(x) \overline{\chi_\xi(x)}\,d\mu(x), \\ p(x) &= \int_{\widehat{\mathcal{X}}} \hat{p}(\xi) \chi_\xi(x)\,d\hat{\mu}(\xi), \quad \xi \in \widehat{\mathcal{X}}.  \end{aligned}
$$
Here, $\hat{\mu}$ denotes the Plancherel measure on $\widehat{\mathcal X}$, i.e., the unique measure that makes the Fourier transform $\mathcal{F}: L^2(\mathcal{X}) \to L^2(\widehat{\mathcal{X}})$ an isometric isomorphism.

\begin{remark}
	The explicit form of the Plancherel measure depends on the structure of the group:
	\begin{enumerate}
		\item[(1)] If $\mathcal X$ is compact (e.g., $\mathbb T^d$), then $\widehat{\mathcal X}$ is discrete and $\hat{\mu}$ is the counting measure;
		\item[(2)] If $\mathcal X$ is non-compact (e.g., $\mathbb R^d$), then $\widehat{\mathcal X}$ is continuous and $\hat{\mu}$ is the Haar measure.
	\end{enumerate}
	In all derivations throughout this paper, we use the unified integral notation $\int_{\widehat{\mathcal{X}}} \cdots d\hat{\mu}(\xi)$ without distinguishing between the discrete and continuous cases.
\end{remark}

For the joint space $\mathcal X\times\mathcal Y$, since both $\mathcal X$ and $\mathcal Y$ are locally compact abelian groups, its Fourier structure is given by $\widehat{\mathcal X\times\mathcal Y}= \widehat{\mathcal X}\times\widehat{\mathcal Y}$, with joint frequencies $(\xi,\eta)$ and corresponding characters
$$
\chi_{(\xi,\eta)}(x,y)=\chi_\xi(x)\chi_\eta(y)
$$
In the Euclidean setting, for example, $e^{i\xi x}e^{i\eta y}=e^{i(\xi x+\eta y)}$.

For $\pi \in \mathcal{M}_0(\mathcal{X}\times\mathcal{Y})$, we define the two-dimensional Fourier transform:
$$
\hat{\pi}(\xi,\eta) = \int_{\mathcal{X}\times\mathcal{Y}} \pi(x,y) \overline{\chi_\xi(x)}\overline{\chi_\eta(y)}\,d\mu(x)d\nu(y), \quad (\xi,\eta) \in \widehat{\mathcal{X}}\times\widehat{\mathcal{Y}}.
$$
By Plancherel's theorem, Parseval's identity holds:
$$
\langle \pi_1, \pi_2 \rangle_{L^2(\mu\otimes\nu)} = \langle \hat{\pi}_1, \hat{\pi}_2 \rangle_{L^2(\hat{\mu}\otimes\hat{\nu})}.
$$
In particular, the norm is conserved:
$$
\int_{\mathcal{X}\times\mathcal{Y}} |\pi|^2\,d\mu d\nu = \int_{\widehat{\mathcal{X}}\times\widehat{\mathcal{Y}}} |\hat{\pi}(\xi,\eta)|^2\,d\hat{\mu}(\xi)d\hat{\nu}(\eta).
$$
From the mass conservation condition $\int q = \int \pi = 0$, we similarly obtain the induced zero-frequency conditions $\hat{q}(0)=0$ and $\hat{\pi}(0,0) = 0$. These conditions ensure that the subsequent variational equations do not diverge at $(\xi,\eta)=(0,0)$, since the numerator also vanishes at the origin, thereby rendering division in frequency space well-defined.
\subsection{Laplace Operator}
For a general locally compact abelian group $\mathcal{X}$, there is no natural Riemannian gradient structure $\nabla$, and hence no canonical Laplace–Beltrami operator $\Delta = \operatorname{div} \nabla$. To address this, we adopt the generator approach to abstractly define a second-order operator: we replace the Laplace operator with a translation-invariant self-adjoint operator that satisfies basic physical axioms, and characterize it completely in the Fourier domain via its spectral multiplier $m(\xi)$.

\begin{Dfn}[Generator]\label{def: generator}
	Let $L_{\mathcal{X}}: D(L_{\mathcal{X}}) \subset L^2(\mathcal{X}) \to L^2(\mathcal{X})$ be a linear operator satisfying the following conditions:
	\begin{enumerate}
		\item[(1)] Self-adjointness: $L_{\mathcal{X}}^* = L_{\mathcal{X}}$;
		\item[(2)] Nonnegativity: $\langle p, L_{\mathcal{X}} p \rangle_{L^2} \ge 0$ for all $p \in D(L_{\mathcal{X}})$;
		\item[(3)] Translation invariance: for any translation $\tau_a: x \mapsto x+a$, we have $L_{\mathcal{X}}(\tau_a p) = \tau_a(L_{\mathcal{X}} p)$.
	\end{enumerate}
	We additionally require the generator to satisfy the Markov property, i.e., $m_{\mathcal X}(0)=0$, and the spectral gap condition $m(\xi,\eta) \ge \lambda > 0$ for all $(\xi, \eta) \neq (0,0)$.
\end{Dfn}

By the spectral theory of translation-invariant self-adjoint operators, $L_{\mathcal X}$ is diagonalized in the Fourier domain as a multiplication operator:
$$
\widehat{L_{\mathcal{X}} p}(\xi) = m_{\mathcal{X}}(\xi) \hat{p}(\xi), \quad \xi \in \widehat{\mathcal{X}}
$$
where $m_{\mathcal{X}}: \widehat{\mathcal{X}} \to (0, \infty)$ is the spectral multiplier.

\begin{remark}
	In fact, the Fourier transform procedure uses only the third property of $L_{\mathcal{X}}$, since any translation-invariant linear operator is necessarily a convolution operator which is naturally diagonalized under the Fourier basis. The full set of defining properties, however, is included to better mimic the behavior of the Laplace operator. In other words, $L_{\mathcal X}$ serves as the mechanism of diffusion and equilibration in the space of capital flows, and the spectral multiplier $m(\xi)$ measures the recovery intensity of risk/liquidity perturbations at different scales. The flatness measure $\mathcal{H}(\pi)$ defined in this paper penalizes imbalances in the joint capital flow distribution across different spatial scales, where high frequencies $m(\xi,\eta)\gg1$ correspond to locally sharp fluctuations and liquidity imbalances.
\end{remark}

\begin{remark}[Fourier Coefficients]\label{remark: FourierCoefficient}
	The Fourier basis functions are the eigenfunctions of the generator $L_{\mathcal X}\chi_\xi=m_{\mathcal X}(\xi)\chi_\xi$; hence, for any function $f(x)=\int\hat f(\xi)\chi_\xi(x)d\xi$, the operator of $L_{\mathcal X}$ acts as:
	$$
	L_{\mathcal X}f=L_{\mathcal X}\int\hat f(\xi)\chi_\xi(x)d\xi=\int\hat f(\xi)L_{\mathcal X}\chi_\xi(x)d\xi=\int m_{\mathcal X}(\xi)\hat f(\xi)\chi_\xi(x)d\xi
	$$
	That is, the Fourier coefficients are:
	$$
	\widehat{L_{\mathcal X}f}(\xi)=m_{\mathcal X}(\xi)\hat f(\xi).
	$$
	In particular, the zero frequency satisfies $m_{\mathcal X}(0)=0$ where no diffusion is generated. The operator acts nontrivially only on the nonzero oscillatory modes, i.e., $L_{\mathcal X}\mathbf1=0$.
\end{remark}

To ensure $L_{\mathcal{X}} p \in L^2(\mathcal{X})$, and to satisfy $\widehat{L_{\mathcal{X}} p}(\xi) = m(\xi)\hat{p}(\xi)$ together with Plancherel's theorem, we define the operator domain $D(L_{\mathcal{X}})$:
$$
D(L_{\mathcal{X}}) = \left\{ p \in L^2(\mathcal{X}) : \int_{\widehat{\mathcal{X}}} m_{\mathcal{X}}(\xi)^2 |\hat{p}(\xi)|^2\,d\hat{\mu}(\xi) < \infty \right\}.
$$
Combined with the zero-mean condition, this yields the physical state space on which the generator acts:
$$
\mathcal{M}_0(\mathcal{X}) \cap D(L_{\mathcal{X}})
= \left\{ p \in L^2(\mathcal{X}) : \hat{p}(0) = 0,\; \int_{\widehat{\mathcal{X}}} m_{\mathcal{X}}(\xi)^2 |\hat{p}(\xi)|^2\,d\hat{\mu}(\xi) < \infty \right\}.
$$

The following two examples illustrate that if $\mathcal X$ happens to possess a differential geometric structure, then $L_{\mathcal X}$ should degenerate to $-\Delta_{\mathcal X}$.

\begin{example}[Riemannian Special Case]
	If $\mathcal{X}$ is an abelian Lie group equipped with a left-invariant Riemannian metric, then we take $L_{\mathcal{X}} = -\Delta_{\mathcal{X}}$ (the negative Laplace–Beltrami operator). In this case, the spectral multiplier is given by
	$$
	\begin{aligned}
		m_{\mathcal{X}}(\xi) &= |\xi|^2,\\  \widehat{\Delta_{\mathcal{X}} p}(\xi) &= -|\xi|^2 \hat{p}(\xi) 
	\end{aligned}
	$$
	where $|\xi|$ denotes the norm on the dual group. For example:
	\begin{enumerate}
		\item[(1)] $\mathcal{X} = \mathbb{R}^d$: $\Delta = \sum_{i=1}^d \partial_i^2$, $\widehat{\Delta p}(\xi) = -|\xi|^2 \hat{p}(\xi)$.
		\item[(2)] $\mathcal{X} = \mathbb{T}^d$: $\xi \in \mathbb{Z}^d$, $\widehat{\Delta p}(\xi) = -|\xi|^2 \hat{p}(\xi)$, where $|\xi|^2 = \xi_1^2 + \cdots + \xi_d^2$.
	\end{enumerate}
\end{example}

\begin{example}[Laplace Operator in the Distributional Sense]
	If $p \notin D(\Delta)$ and $\mathcal X$ is an abelian Lie group, then for $\phi \in C_c^\infty(\mathcal{X})$, we define $\langle \Delta p, \phi \rangle = \langle p, \Delta \phi \rangle$. In this case, $\widehat{\Delta p}(\xi) = -m(\xi)\hat{p}(\xi)$ still holds in the distributional sense.
\end{example}

In the Euclidean setting, $\Delta_{x,y} = \Delta_x + \Delta_y$, and the total generator is the sum of two independent generators, meaning that the system evolves independently in the two directions. We expect the joint generator $L_{\mathcal{X}\times\mathcal{Y}}$ to possess a similar property, where $L_{\mathcal{X}} \otimes I$ denotes an operator acting only in the $x$-direction, independently of the $y$-direction. Taking the transport matrix as an example, $(L_X \otimes I)$: for each row (fixing a certain $y$), one applies $L_{\mathcal{X}}$ to all column data in that row, while different columns do not interact with one another.

\begin{Dfn}[Joint Generator]
	For the product space $\mathcal{X} \times \mathcal{Y}$, define the joint generator
	$$
	L_{\mathcal{X}\times\mathcal{Y}} = L_{\mathcal{X}} \otimes I + I \otimes L_{\mathcal{Y}}.
	$$
	Here $L_{\mathcal X}\otimes I$ denotes that
	$$
	(L_{\mathcal X}\otimes I)\pi(x,y)=L_{\mathcal X}[\pi(\cdot,y)](x).
	$$
\end{Dfn}

\begin{Lem}\label{lem: UnitedGeneratorFourierTrans}
	In the Fourier domain,
	$$
	\widehat{L_{\mathcal{X}\times\mathcal{Y}} \pi}(\xi,\eta) = \big( m_{\mathcal{X}}(\xi) + m_{\mathcal{Y}}(\eta) \big) \hat{\pi}(\xi,\eta).
	$$
\end{Lem}

\begin{Cor}
	In the Riemannian special case,
	$$
	\widehat{\Delta_{x,y}\pi}(\xi,\eta) = -\big( |\xi|^2 + |\eta|^2 \big) \hat{\pi}(\xi,\eta).
	$$
\end{Cor}

\begin{Lem}[Fourier Characterization of Marginal Projections]\label{lem: maginalPorjFourierRepret}
	For any $\xi \in \widehat{\mathcal{X}}$ and $\eta \in \widehat{\mathcal{Y}}$,
	$$
	\begin{aligned}
		\widehat{P_{\mathcal X}\pi}(\xi) &= \hat{\pi}(\xi, 0)\\
		\widehat{P_{\mathcal Y}\pi}(\eta) &= \hat{\pi}(0, \eta).
	\end{aligned}
	$$
\end{Lem}
\begin{proof}
	We prove the case for $P_{\mathcal X}\pi$. By the definition of the Fourier transform $\widehat{P_{\mathcal X}\pi}(\xi)= \int_{\mathcal{X}} (P_{\mathcal X}\pi)(x) \overline{\chi_\xi(x)}\,d\mu(x)$, and substituting the definition of the marginal projection $(P_{\mathcal X}\pi)(x) = \int_{\mathcal{Y}} \pi(x,y)\,d\nu(y)$, we obtain
	$$
	\widehat{P_{\mathcal X}\pi}(\xi)
	= \int_{\mathcal{X}}\int_{\mathcal{Y}} \pi(x,y) \overline{\chi_\xi(x)}\,d\nu(y)d\mu(x).
	$$
	Applying Fubini's theorem to interchange the order of integration, and noting that $\chi_0(y) \equiv 1$ (when $\eta=0$), we have $\overline{\chi_\eta(y)}\big|_{\eta=0} = 1$. Hence,
	$$
	\begin{aligned}
		\widehat{P_{\mathcal X}\pi}(\xi) &=\int_{\mathcal{X}\times\mathcal{Y}}\pi(x,y)\overline{\chi_\xi(x)}d\mu d\nu \\&= \int_{\mathcal{X}\times\mathcal{Y}} \pi(x,y) \overline{\chi_\xi(x)} \cdot 1 \, d\mu(x)d\nu(y) \\&= \int_{\mathcal{X}\times\mathcal{Y}} \pi(x,y) \overline{\chi_\xi(x)} \overline{\chi_0(y)}\,d\mu(x)d\nu(y) \\&= \hat{\pi}(\xi, 0).
	\end{aligned}
	$$
	The proof for $\widehat{P_{\mathcal Y}\pi}(\eta) = \hat{\pi}(0, \eta)$ follows analogously.
\end{proof}

\begin{remark}\label{remark: fixedMarginalDistrib}
	Lemma \ref{lem: maginalPorjFourierRepret} is equivalent to
	$$
	\begin{aligned}\hat{\pi}(\xi, 0) &= \hat{p}(\xi), \\ \hat{\pi}(0, \eta) &= \hat{q}(\eta). \end{aligned}
	$$
	This means that the Fourier transform $\hat{\pi}(\xi,\eta)$ of the joint distribution is completely fixed by the marginals on the two coordinate axes, occupying only the zero-frequency axes. Moreover, when $\mathcal{X}$ and $\mathcal{Y}$ are compact abelian groups, the spectra on $\widehat{\mathcal{X}}$ and $\widehat{\mathcal{Y}}$ are given by discrete counting measures, and we may select frequency slices $\xi=0$ or $\eta=0$ to truncate the energy:
	$$
	\widehat{\mathcal{X}} \times \widehat{\mathcal{Y}}=\big(\{0\} \times \widehat{\mathcal{Y}}\big)\cup\big(\widehat{\mathcal{X}} \times \{0\}\big)\cup\big(\widehat{\mathcal{X}} \setminus \{0\}\big) \times\big(\widehat{\mathcal{Y}} \setminus \{0\}\big).
	$$
\end{remark}

\subsection{Dirichlet Energy and Regularization}
\subsubsection{Energy Space}
\begin{Dfn}[Energy Space]
	The energy space induced by the generator is given by $\mathcal E$, which serves as the natural domain of the Dirichlet energy:
	$$
	\mathcal E = D(L_{\mathcal X\times\mathcal Y}^{1/2}) = \left\{ \pi \in L^2(\mathcal{X}\times\mathcal{Y}) : \int m(\xi,\eta)|\hat\pi(\xi,\eta)|^2\,d\hat\mu d\hat\nu < \infty\right\}. 
	$$
	where $m(\xi,\eta) = m_{\mathcal X}(\xi) + m_{\mathcal Y}(\eta)$. This is the minimal space on which $\langle\pi,L_{\mathcal X\times\mathcal Y}\pi\rangle<\infty$ is finite, i.e., the space of functions with finite energy. 
\end{Dfn}

\begin{remark}\label{remark: halfGeneratorIsSelfAdjointOpe}
	The spectral theorem implies that if $L_{\mathcal X\times\mathcal Y}$ is self-adjoint and $L_{\mathcal X\times\mathcal Y}\geq 0$, then $\sigma(L_{\mathcal X\times\mathcal Y})\subset[0,\infty)$. Moreover, $L_{\mathcal X\times\mathcal Y}^{1/2}:\mathcal E \to L^2(\mathcal{X}\times\mathcal{Y})$ is automatically self-adjoint with $\sigma(L_{\mathcal X\times\mathcal Y}^{1/2})\subset[0,\infty)$.
\end{remark}

Recall the domain of the Laplace operator, i.e., the space on which $L_{\mathcal X\times\mathcal Y}\pi\in L^2$:
$$
D(L_{\mathcal X\times\mathcal Y}) = \left\{\pi\in L^2(\mathcal{X}\times\mathcal{Y}) : \int m(\xi,\eta)^2|\hat\pi(\xi,\eta)|^2\,d\hat\mu d\hat\nu < \infty\right\}.
$$

\begin{Lem}[Relationships between Energy Spaces]\label{lem: RelationshipsBetweenEnergySpaces}
	The relationships between the spaces are given by
	$$
	D(L_{\mathcal X\times\mathcal Y}) \subset \mathcal E = D(L_{\mathcal X\times\mathcal Y}^{1/2}) \subset L^2(\mathcal X\times\mathcal Y)
	$$
\end{Lem}
Consequently, the Dirichlet energy functional need not be confined to the strong solution space $\mathcal{M}_0(\mathcal{X}) \cap D(L_{\mathcal{X}})$ defined earlier; it can be defined in a more relaxed setting, for instance, merely in the quadratic-form sense. Lemma \ref{lem: RelationshipsBetweenEnergySpaces} shows that $\mathcal E$ is precisely the set of all square-integrable functions with finite energy.

\begin{example}[Sobolev Nesting for the Laplace Operator]
	Let $\Omega \subset \mathbb{R}^n$ be bounded and smooth, and let $L = -\Delta$ be equipped with Dirichlet boundary conditions. Then we have the following Sobolev nesting:
	$$
	H^2(\Omega) \cap H_0^1(\Omega) \subset H_0^1(\Omega) \subset L^2(\Omega)
	$$
\end{example}

We prove the first major proposition that will be frequently used in the remainder of this paper.

\begin{Prop}\label{prop: EnergySpaceIsHilbert}
	$\left(\mathcal E,\|\pi\|_{\mathcal E}\right)$ is a Hilbert space, where $\|\pi\|_{\mathcal E} = \left(\|\pi\|_{L^2}^2 + \|L^{1/2}\pi\|_{L^2}^2\right)^{1/2}$ denotes the graph norm induced by the inner product $\langle u,v\rangle_{\mathcal E}=\langle u,v\rangle_{L^2}+\langle L^{1/2}u,L^{1/2}v\rangle_{L^2}$.
\end{Prop}
\begin{proof}
	Since $L_{\mathcal X\times\mathcal Y}^{1/2}$ is a linear operator, it is straightforward to verify that $\langle u,v\rangle_{\mathcal E}=\langle u,v\rangle_{L^2}+\langle L^{1/2}u,L^{1/2}v\rangle_{L^2}$ defines an inner product.
	
	For completeness, we first characterize $\mathcal E=D(L_{\mathcal X\times\mathcal Y}^{1/2})$. The spectral theorem states that any self-adjoint operator on a Hilbert space can be represented as an integral with respect to its spectral projections. From Lemma \ref{lem: RelationshipsBetweenEnergySpaces} and Remark \ref{remark: halfGeneratorIsSelfAdjointOpe}, we know that $L_{\mathcal X\times\mathcal Y}^{1/2}$ is a nonnegative self-adjoint operator on the Hilbert space $L^2(\mathcal{X}\times\mathcal{Y})$ (in the following derivation, we abbreviate $L^{1/2}_{\mathcal X\times\mathcal Y}\pi$ as $L^{1/2}\pi$). Applying the spectral theorem on $L^2(\mathcal{X}\times\mathcal{Y})$, we have $L^{1/2}_{\mathcal X\times\mathcal Y} = \int_0^\infty \sqrt{\lambda} \, dE_\lambda$. For $\pi \in \mathcal D(L_{\mathcal X\times\mathcal Y}^{1/2})$, $L^{1/2}\pi\in L^2(\mathcal{X}\times\mathcal{Y})$ satisfies
	$$
	\begin{aligned}L^{1/2}\pi &= \int_0^\infty \sqrt{\lambda} \, dE_\lambda \pi,\\ \|L^{1/2}\pi\|^2&=\int_0^\infty\lambda d\|E_\lambda \pi\|^2. \end{aligned}
	$$
	which further yields
	$$
	\pi\in D(L^{1/2})\iff\int_0^\infty\lambda d\|E_\lambda \pi\|^2<\infty.
	$$
	From Lemma \ref{lem: RelationshipsBetweenEnergySpaces}, we have $\pi \in \mathcal D(L_{\mathcal X\times\mathcal Y}^{1/2})\subset L^2(\mathcal{X}\times\mathcal{Y})$, so $\|\pi\|^2<\infty$, and consequently $\|\pi\|^2+\|L^{1/2}\pi\|^2<\infty$. This shows that $D(L^{1/2}_{\mathcal X\times\mathcal Y})$ can be precisely characterized by the spectral representation as
	$$
	D(L^{1/2}_{\mathcal X\times\mathcal Y})=\{\pi\in L^2(\mathcal{X}):\|\pi\|^2+\|L^{1/2}\pi\|^2<\infty\}.
	$$
	Since $L_{\mathcal X\times\mathcal Y}$ is translation-invariant, we have $\widehat{L_{\mathcal X\times\mathcal Y}\pi} = m(\xi,\eta)\hat \pi(\xi,\eta)$ with $m(\xi,\eta)\ge0$. By the spectral functional calculus, $L_{\mathcal X\times\mathcal Y}^{1/2}$ is also translation-invariant, and from $L_{\mathcal X\times\mathcal Y}^{1/2} = \mathcal F^{-1} M_{\sqrt{m}} \mathcal F$ we obtain $\mathcal F(L^{1/2}\pi)= \sqrt{m(\xi)}\,\hat \pi(\xi)$, i.e.,
	$$
	\widehat{L^{1/2}\pi} = \sqrt{m(\xi,\eta)}\hat \pi(\xi,\eta).
	$$
	Thus, $\pi\in D(L_{\mathcal X\times\mathcal Y}^{1/2})$ is equivalent to $\sqrt m\hat \pi\in L^2(\mathcal{X}\times\mathcal{Y})$, i.e., $\int m(\xi,\eta)|\hat \pi(\xi,\eta)|^2 d\hat\mu d\hat\nu < \infty$. By Parseval's formula, we have $\|\pi\|_2^2 = \int |\hat \pi|^2$ and $\|L^{1/2}\pi\|_2^2 = \int m|\hat \pi|^2$, i.e.,
	$$
	\|\pi\|_{\mathcal E}^2 = \int (1+m(\xi,\eta))|\hat \pi(\xi,\eta)|^2 d\hat\mu d\hat\nu
	$$
	
	Let $w(\xi,\eta)=1+m(\xi,\eta)\ge1$, and define the weighted space $L^2_w = \left\{ f : \int |f|^2 w \, d\hat\mu d\hat\nu < \infty \right\}$. This means that $\|\pi\|_{\mathcal E}^2 = \int w|\hat \pi|^2 \, d\hat\mu d\hat\nu = \|\hat \pi\|_{L^2_w}^2$, i.e., $\|\pi\|_{\mathcal E} = \|\hat \pi\|_{L^2_w}$.
	We now verify that $\mathcal F: \mathcal E \to L^2_w$ ($\pi \mapsto \hat \pi$) is bijective. If $\hat \pi = \hat r$, by the unitarity of the Fourier transform we have $\pi = r$, so $\mathcal F$ is injective. For any $f \in L^2_w$, since $w \ge 1$, we have $f \in L^2$. Since Fourier transform $\mathcal F: L^2 \to L^2$ is unitary, there exists a unique $\pi = \mathcal F^{-1}f \in L^2$. From $\hat \pi = f$, we obtain
	$$
	\int m|\hat \pi|^2 \, d\hat\mu d\hat\nu = \int m|f|^2 \, d\hat\mu d\hat\nu \le \int (1+m)|f|^2 \, d\hat\mu d\hat\nu < \infty.
	$$
	i.e., $\hat \pi$ satisfies the Fourier characterization of $\mathcal E$, so $\pi \in \mathcal E$. Therefore $\mathcal F$ is surjective. In summary, $\mathcal F$ is bijective. Hence,
	$$
	\mathcal E \cong L^2_w.
	$$
	The space $L^2_w$ is an $L^2$-space equipped with the inner product $\langle f,g\rangle_w = \int f\bar g \, w \, d\hat\mu d\hat\nu$. Since the measure satisfies $d\rho=w\,d\hat\mu d\hat\nu$, we have $L^2_w=L^2(\rho)$. By the Riesz–Fischer theorem, $L^2_w$ is complete.
	Now let $\pi_n$ be a Cauchy sequence in $\mathcal E$, i.e., $\|\pi_n-\pi_m\|_{\mathcal E}\rightarrow0$. After Fourier transform, $\hat \pi_n$ satisfies
	$$
	\int (1+m)|\hat \pi_n-\hat \pi_m|^2 \rightarrow0.
	$$
	i.e., $\{\hat \pi_n\}$ is a Cauchy sequence in $L^2_w$. Since $L^2_w$ is complete, there exists $f\in L^2_w$ such that $\hat \pi_n\rightarrow f$. By $\mathcal E \cong L^2_w$, there exists a unique $\pi=\mathcal F^{-1}f$, with $\pi\in \mathcal E$. Moreover,
	$$
	\|\pi_n - \pi\|_{\mathcal E} = \|\hat \pi_n - f\|_{L^2_w} \to 0.
	$$
	Thus $\{\pi_n\}$ converges to $\pi \in \mathcal E$. Therefore $\mathcal E$ is complete, i.e., $(D(L_{\mathcal X\times\mathcal Y}^{1/2}), \langle\cdot,\cdot\rangle_{\mathcal E})$ is a Hilbert space.
\end{proof}

\begin{remark}\label{remark: HalfOperatorFourierTrans}
	The proof of Proposition \ref{prop: EnergySpaceIsHilbert} contains the following Fourier-domain characterization:
	$$
	\widehat{L^{1/2}_{\mathcal X\times\mathcal Y}\pi} = \sqrt{m(\xi,\eta)}\hat \pi(\xi,\eta).
	$$
\end{remark}

\begin{remark}\label{remark: DenseResults}
	The space on which the subsequent entropy regularization acts is $\mathcal E = D(L_{\mathcal X\times\mathcal Y}^{1/2})$. The existence of a minimizer for the optimization problem relies on the reflexivity of $\left(\mathcal E,\|\cdot\|_{\mathcal E}\right)$. In fact, $D(L_{\mathcal X\times\mathcal Y})$, when equipped with the graph norm $\|\pi\|_{D(L)}=\left(\|\pi\|_{L^2}^2+\|L\pi\|_{L^2}^2\right)^{1/2}$ is also a Hilbert space, whose completeness is guaranteed by the closedness of the nonnegative self-adjoint operator $L_{\mathcal X\times\mathcal Y}$. It is worth noting that the Hilbert properties of both spaces are induced by the self-adjoint operator itself, rather than by the underlying locally compact abelian group structure. Furthermore, by introducing the spectral truncation operator $P_n := E([0, n]) = \int_0^n dE_\lambda$, it is easy to check that $D(L_{\mathcal X\times\mathcal Y})\stackrel{dense}{\subset}\mathcal E$ holds under the norm $\|\cdot\|_{\mathcal E}$, i.e., $\overline{D(L)}^{\|\cdot\|_{\mathcal E}} = \mathcal E$. This conclusion is contained in the proof of the Density Lemma \ref{lem: DenseLemma}. By Lemma \ref{lem: RelationshipsBetweenEnergySpaces}, this is equivalent to showing that states with infinite energy can be approximated by states with finite energy. The density lemma further states that when $\mathcal X$ and $\mathcal Y$ are compact, $\mathcal E_0 \stackrel{dense}{\subset} \mathcal M_0$ holds in the $L^2$-sense, i.e., $\overline{\mathcal E_0}^{L^2}=\mathcal M_0$. 
\end{remark}

\subsubsection{Dirichlet Energy}
\begin{Dfn}[Dirichlet Energy]
	For $\pi \in \mathcal E = D(L^{1/2}_{\mathcal X\times\mathcal Y}) \subset L^2(\mathcal X\times\mathcal Y)$, the Dirichlet energy is defined as
	$$\mathcal H(\pi) = \frac12 \|L_{\mathcal X\times\mathcal Y}^{1/2}\pi\|_{L^2}^2.$$
\end{Dfn}

\begin{remark}[Well-definedness]
	The quadratic form $\frac12 \langle \pi, L\pi \rangle_{L^2}$ is defined only on $\pi\in D(L_{\mathcal X\times\mathcal Y})$. For $\pi\in\mathcal E\setminus D(L_{\mathcal X\times\mathcal Y})$, the inner product need not exist. Note that for any $\pi \in D(L_{\mathcal X\times\mathcal Y})$, since $L = (L^{1/2})^*L^{1/2}$ and $L^{1/2}$ is self-adjoint, we have
	$$
	\langle \pi, L\pi \rangle_{L^2}= \langle L^{1/2}\pi, L^{1/2}\pi \rangle_{L^2}= \|L^{1/2}\pi\|_{L^2}^2.
	$$
	However, the right-hand side requires only $L^{1/2}\pi \in L^2_{\mathcal X\times\mathcal Y}$, i.e., $\pi \in \mathcal E$. Therefore, the Dirichlet energy $\mathcal H(\pi) = \frac12 \|L_{\mathcal X\times\mathcal Y}^{1/2}\pi\|_{L^2}^2$ is the extension of $\frac12 \langle \pi, L\pi \rangle_{L^2}$ to $\mathcal E$. The density of $D(L_{\mathcal X\times\mathcal Y})$ in $\mathcal E$ and the completeness of $\left(\mathcal E,\|\cdot\|_{\mathcal E}\right)$ guarantee the existence and uniqueness of the closed continuous quadratic extension $\mathcal H(\pi)$. Hence $\mathcal H(\pi)$ is well-defined.
\end{remark}

From Remark \ref{remark: HalfOperatorFourierTrans} and Parseval's identity, the Fourier representation of the Dirichlet energy is given by
$$
\begin{aligned} \mathcal H(\pi) &= \frac12 \|L_{\mathcal X\times\mathcal Y}^{1/2}\pi\|_{L^2}^2 =\frac12 \|\widehat{L_{\mathcal X\times\mathcal Y}^{1/2}\pi}\|_{L^2}^2   \\ &=\frac12 \int_{\widehat{\mathcal X}\times\widehat{\mathcal Y}}\left| \widehat{L_{\mathcal X\times\mathcal Y}^{1/2}\pi}(\xi,\eta) \right|^2\, d\hat{\mu}d\hat{\nu}  \\&= \frac12 \int_{\widehat{\mathcal X}\times\widehat{\mathcal Y}}m(\xi,\eta)|\hat{\pi}(\xi,\eta)|^2\, d\hat{\mu}d\hat{\nu}. \end{aligned}
$$
that is,
\begin{equation}\label{eq: DirichletEnerFourierRepre}
	\mathcal H(\pi)= \frac12 \int_{\widehat{\mathcal X}\times\widehat{\mathcal Y}}\big( m_{\mathcal X}(\xi) + m_{\mathcal Y}(\eta) \big)|\hat{\pi}(\xi,\eta)|^2\, d\hat{\mu}d\hat{\nu},\quad \pi \in \mathcal E.
\end{equation}

\begin{example}
	In the Riemannian special case where $m_{\mathcal X} = |\xi|^2$ and $m_{\mathcal Y} = |\eta|^2$, equation \ref{eq: DirichletEnerFourierRepre} degenerates to the classical Dirichlet energy:
	$$
	\begin{aligned} \mathcal{H}(\pi) &= \frac{1}{2}\int_{\mathcal{X}\times\mathcal{Y}} |\nabla_{x,y}\pi(x,y)|^2\,d\mu(x)d\nu(y)
		\\&= \frac{1}{2}\int_{\widehat{\mathcal{X}}\times\widehat{\mathcal{Y}}} (|\xi|^2+|\eta|^2)|\hat{\pi}(\xi,\eta)|^2\,d\hat{\mu}d\hat{\nu}. \end{aligned}
	$$
\end{example}

To ultimately prove Proposition \ref{prop: optimalEstimationExistUniq}, we establish the following lemmas to understand the features of the signed transport problem after the introduction of the regularization term, as well as the structure of feasible solutions in the relevant function spaces. 

\begin{Lem} \label{lem: N-EspaceHilbert}
	$(\mathcal M_0 \cap \mathcal E,\|\cdot\|_\mathcal E)$ is a Hilbert space.
\end{Lem}
\begin{proof}
	Define the linear functional $T: \mathcal E \to \mathbb R$ by $T(\pi) = \int \pi= \langle \pi, \mathbf{1} \rangle_{L^2}$. Applying Cauchy–Schwarz yields
	$$
	|T(\pi)| \le \|\pi\|_{L^2} \|\mathbf{1}\|_{L^2} \le C\|\pi\|_{\mathcal E}.
	$$
	Thus $T$ is bounded (hence continuous). The space $\mathcal M_0 = \ker T := \{ \pi \in \mathcal E : T(\pi) = 0 \}$ is the kernel of the continuous linear functional $T$, and the singleton $\{0\}$ is closed in $\mathbb R$. Therefore, $\mathcal M_0$ is a closed subspace of $\mathcal E$. Since $\mathcal M_0 \cap \mathcal E$ is complete under the inherited norm, it is a Hilbert space.
\end{proof}

\begin{Lem}[Convexity of Dirichlet Energy]\label{lem: ConvexityOfDirletEnergy}
	$\mathcal H(\pi) = \frac12 \|L_{\mathcal X\times\mathcal Y}^{1/2}\pi\|_{L^2}^2$ is strictly convex on $\mathcal M_0 \cap \mathcal E$. 
\end{Lem}

\begin{Lem}[Coercivity]\label{lem: Coercivity of Operator}
	$\|\pi\|_{\mathcal E}^2 \le C \|L_{\mathcal X\times\mathcal Y}^{1/2}\pi\|_{L^2}^2$, $\forall \pi \in \mathcal E_0$.
\end{Lem}
\begin{proof}
	For any $\pi \in \mathcal E_0$, we have $\hat{\pi}(0,0) = 0$. Then, by the spectral gap condition, we obtain
	$$
	\|\pi\|_{L^2}^2=\int_{\widehat{\mathcal X}\times\widehat{\mathcal Y}}|\hat{\pi}(\xi,\eta)|^2\, d\hat{\mu}(\xi)d\hat{\nu}(\eta)\le\frac{1}{\lambda}\int_{\widehat{\mathcal X}\times\widehat{\mathcal Y}}m(\xi,\eta)|\hat{\pi}(\xi,\eta)|^2\, d\hat{\mu}(\xi)d\hat{\nu}(\eta)=\frac{2}{\lambda}\mathcal H(\pi).
	$$
	Hence,
	$$
	\|\pi\|_{\mathcal E}^2 = \|\pi\|_{L^2}^2 + \|L^{1/2}\pi\|_{L^2}^2
	\le \left(1 + \frac{2}{\lambda}\right) \|L^{1/2}\pi\|_{L^2}^2.
	$$
	That is, there exists a constant $C > 0$ such that
	$$
	\|\pi\|_{\mathcal E}^2 \le C \|L^{1/2}\pi\|_{L^2}^2, \quad \forall \pi \in \mathcal E_0.
	$$
\end{proof}

\begin{remark} \label{remark: CoercivityOfOperator}
	Coercivity implies that, on $\mathcal E_0 = \mathcal M_0 \cap \mathcal E$, the graph norm is equivalent to the energy seminorm. The proof also contains $\|\pi\|_{L^2}^2 \le \frac{2}{\lambda} \mathcal H(\pi)$. Moreover, if $\pi$ is $D(L_{\mathcal X\times\mathcal Y})-$regular, i.e., $\pi \in D(L_{\mathcal X\times\mathcal Y})$, then
	$$
	\|\pi\|_{\mathcal E}^2 \le C \|L^{1/2}\pi\|_{L^2}^2 = C \langle L^{1/2}\pi, L^{1/2}\pi \rangle_{L^2} = C\langle \pi, L\pi \rangle_{L^2} 
	$$
\end{remark}

\subsubsection{Regularization}
Let $\mathcal E_0 := \mathcal M_0 \cap \mathcal E= \left\{ \pi \in D(L^{1/2}_{\mathcal X\times\mathcal Y}) : \hat{\pi}(0,0) = 0 \right\}$. Define the feasible set:
$$
\Pi_{\mathcal E_0}(p,q) := \left\{ \pi \in \mathcal E_0 : P_{\mathcal X}\pi = p,\; P_{\mathcal Y}\pi = q \right\}.
$$
Clearly, $\Pi_{\mathcal E_0}(p,q)\neq \varnothing$, since the independent coupling $\pi_{Ind}\in \Pi_{\mathcal E_0}(p,q)$ satisfies the constraints.

The entropy-regularized formulation of the signed transport problem is given by
$$
\inf_{\pi \in \Pi_{\mathcal E_0}(p,q)} J_\varepsilon(\pi) := \int_{\mathcal{X}\times\mathcal{Y}} c(x,y)\,d\pi(x,y) + \varepsilon \mathcal H(\pi),
$$
where $c \in L^2(\mathcal X\times\mathcal Y)$ and $\langle c, \pi \rangle_{L^2}= \int_{\mathcal X\times\mathcal Y} c(x,y) \pi(x,y) \, d\mu d\nu$. Denote the optimal solution of the problem by $\pi_\varepsilon^* := \arg\min_{\pi \in \Pi_{\mathcal E_0}(p,q)}J_\varepsilon(\pi)$.

\begin{Lem}[Structure of Feasible Solutions]\label{lem: StructrueOfFeasibleSol}
	The set $\Pi_{\mathcal E_0}(p,q)$ is a closed convex subset of $\mathcal E_0$.
\end{Lem}
\begin{proof}
	For any $\pi_1, \pi_2 \in \Pi_{\mathcal E_0}(p,q)$ and $t \in [0,1]$, let $\pi_t = t\pi_1 + (1-t)\pi_2$. Then
	$$
	P_{\mathcal X}\pi_t = t p + (1-t)p = p, \quad P_{\mathcal Y}\pi_t = q,
	$$
	and $\pi_t \in \mathcal E_0$ (since $\mathcal E_0$ is a linear space). Hence $\Pi_{\mathcal E_0}(p,q)$ is convex.
	The marginal projection operators $P_{\mathcal X}: \mathcal E_0 \to L^2(\mathcal X)$ and $P_{\mathcal Y}: \mathcal E_0 \to L^2(\mathcal Y)$ are continuous linear operators (the proof is given below). Since the singletons $\{p\}$ and $\{q\}$ are closed sets, we have
	$$
	\Pi_{\mathcal E_0}(p,q) = P_{\mathcal X}^{-1}(\{p\}) \cap P_{\mathcal Y}^{-1}(\{q\}),
	$$
	which is closed.
	We now prove the continuity of the marginal projections. For $P_{\mathcal X}$, by Cauchy–Schwarz,
	$$
	\begin{aligned}|P_{\mathcal X}\pi(x)| &= \left| \int_{\mathcal Y} \pi(x,y)\,d\nu(y) \right|
		\\&\le \|\pi(x,y)\|_{L^2(\mathcal Y)} \|1\|_{L^2(\mathcal Y)}=\sqrt{\nu(\mathcal Y)} \left( \int_{\mathcal Y} |\pi(x,y)|^2\,d\nu(y) \right)^{1/2}, \end{aligned}
	$$
	where $\|1\|_{L^2(\mathcal Y)}^2 = \int_{\mathcal Y} 1^2 \, d\nu(y) = \nu(\mathcal Y)$. Squaring and integrating yields
	$$
	\|P_{\mathcal X}\pi\|_{L^2(\mathcal X)}^2 \le \nu(\mathcal Y) \|\pi\|_{L^2(\mathcal X\times\mathcal Y)}^2 \le \nu(\mathcal Y) \|\pi\|_{\mathcal E}^2.
	$$
	Therefore, $P_{\mathcal X}$ is continuous. The continuity of $P_{\mathcal Y}$ follows analogously.
\end{proof}

\begin{remark}\label{remark: ProjOperatorIsConti}
	The proof of Lemma \ref{lem: StructrueOfFeasibleSol} contains the continuity of the marginal projection operators $P_{\mathcal X}$ and $P_{\mathcal Y}$. For any $\pi_1,\pi_2\in\Pi_{\mathcal E_0}(p,q)$, we have $P_X(\pi_1+\pi_2)=2p$, which implies $\pi_1+\pi_2\notin\Pi_{\mathcal E_0}(p,q)$. Consequently, $\Pi_{\mathcal E_0}(p,q)$ is not a linear space, and hence not a Hilbert space.
\end{remark}

\begin{Lem} \label{lem: RestritionFunctionIsConti}
	The embedding $\mathcal E_0 \hookrightarrow L^2(\mathcal{X}\times\mathcal{Y})$ is continuous, and the linear functional $T_c(\pi): L^2(\mathcal{X}\times\mathcal{Y})\to \mathbb{R}$ ($\pi\mapsto \langle c, \pi \rangle_{L^2}$) is continuous on $\mathcal E_0$, i.e., $T_c|_{\mathcal E_0}=\langle \pi, \eta_c \rangle_{\mathcal E_0}$ is a bounded linear functional.
\end{Lem}
\begin{proof}
	Clearly $\mathcal E_0 \subset L^2$. It suffices to show that for all $\pi \in \mathcal E_0$, $\|\pi\|_{L^2} \le C \|\pi\|_{\mathcal E_0}$. By the definition of the graph norm,
	$$
	\|\pi\|_{L^2}^2 \le \|\pi\|_{L^2}^2 + \|L^{1/2}\pi\|_{L^2}^2 = \|\pi\|_{\mathcal E_0}^2.
	$$
	Taking square roots yields $\|\pi\|_{L^2} \le \|\pi\|_{\mathcal E_0}$, so $\mathcal E_0 \hookrightarrow L^2(\mathcal{X}\times\mathcal{Y})$ is a continuous embedding.
	We now prove that $T_c|_{\mathcal E_0}$ is a continuous linear functional on $\mathcal E_0$. For any $\pi_1, \pi_2 \in \mathcal E_0$ and $\alpha, \beta \in \mathbb R$,
	$$
	T_c(\alpha\pi_1 + \beta\pi_2) = \langle c, \alpha\pi_1 + \beta\pi_2 \rangle_{L^2}
	= \alpha\langle c, \pi_1 \rangle_{L^2} + \beta\langle c, \pi_2 \rangle_{L^2}.
	$$
	By the Cauchy–Schwarz inequality (in $L^2(\mathcal{X}\times\mathcal{Y})$), we have $|T_c(\pi)| = |\langle c, \pi \rangle_{L^2}| \le \|c\|_{L^2} \|\pi\|_{L^2}$. Since $\mathcal E_0 \hookrightarrow L^2$ is a continuous embedding, i.e., $\|\pi\|_{L^2} \le \|\pi\|_{\mathcal E_0}$, it follows that
	$$
	|T_c(\pi)| \le \|c\|_{L^2} \|\pi\|_{\mathcal E_0}.
	$$
	Therefore, $T_c|_{\mathcal E_0}$ is a bounded linear functional with bound $\|c\|_{L^2}<\infty$ where $c \in L^2(\mathcal X\times\mathcal Y)$.
\end{proof}

\begin{Prop}\label{prop: optimalEstimationExistUniq}
	$\pi_\varepsilon^*$  exists and is unique.
\end{Prop}
\begin{proof}
	We first prove existence. Let $a := \inf_{\pi \in \Pi_{\mathcal E_0}(p,q)} J_\varepsilon(\pi)$. Then there exists a minimizing sequence $\{\pi_n\} \subset \Pi_{\mathcal E_0}(p,q)$ such that $J_\varepsilon(\pi_n) \to a$ (for $n \in \mathbb{N}$, with $J_\varepsilon(\pi_n) \le a + \frac{1}{n}$). Let $M := \max\{a+1, J_\varepsilon(\pi_1), \ldots, J_\varepsilon(\pi_N)\}>0$; then we have $J_\varepsilon(\pi_n) \le M$.
	Since $c \in L^2$, by Cauchy–Schwarz we obtain $\langle c, \pi_n \rangle \ge -\|c\|_{L^2} \|\pi_n\|_{L^2}$, and hence
	$$
	J_\varepsilon(\pi_n) = \langle c, \pi_n \rangle + \varepsilon \mathcal H(\pi_n)
	\ge -\|c\|_{L^2} \|\pi_n\|_{L^2} + \varepsilon \mathcal H(\pi_n).
	$$
	From $J_\varepsilon(\pi_n) \le M$ we get $\varepsilon \mathcal H(\pi_n) - \|c\|_{L^2} \|\pi_n\|_{L^2} \le M$. Using coercivity (Remark \ref{remark: CoercivityOfOperator}) to control $\|\pi_n\|_{L^2}$ via $\|\pi_n\|_{L^2} \le \sqrt{\frac{2}{\lambda}} \sqrt{\mathcal H(\pi_n)}$, we obtain
	$$
	\varepsilon \mathcal H(\pi_n)- \|c\|_{L^2} \sqrt{\frac{2}{\lambda}} \sqrt{\mathcal H(\pi_n)}\le M.
	$$
	This is a quadratic inequality in $\sqrt{\mathcal H(\pi_n)}\geq 0$, which ensures that $\mathcal H(\pi_n)$ cannot tend to infinity; otherwise the right-hand side would be exceeded. Let $\mathcal H(\pi_n)\leq B$. Again by Lemma \ref{lem: Coercivity of Operator}, we have
	$$
	\|\pi_n\|_{\mathcal E}^2 \le C \|L_{\mathcal X\times\mathcal Y}^{1/2}\pi_n\|_{L^2}^2 = 2C \mathcal H(\pi_n) \le 2CB.
	$$
	Therefore, $\{\pi_n\}$ is bounded in $\mathcal E_0$.
	
	By the reflexivity of the Hilbert space $\mathcal E_0$ and the Banach–Alaoglu theorem (Lemma \ref{lem: N-EspaceHilbert}), every bounded sequence has a weakly convergent subsequence. Hence, there exists a subsequence (still denoted by $\{\pi_n\}$) and an element $\pi^* \in \mathcal E_0$ such that
	$$
	\pi_n \rightharpoonup \pi^*.
	$$
	i.e., $\langle \pi_n, \phi \rangle_{\mathcal E_0} \to \langle \pi^*, \phi \rangle_{\mathcal E_0}$, $\forall \phi \in \mathcal E_0$.
	This limit preserves the marginal constraints. The marginal projection operators $P_{\mathcal X}, P_{\mathcal Y}: \mathcal E_0 \to L^2$ are continuous linear operators (Remark \ref{remark: ProjOperatorIsConti}) and hence weakly continuous. Taking weak limits in $P_{\mathcal X}\pi_n = p$ and $P_{\mathcal Y}\pi_n = q$, we obtain
	$$
	\begin{aligned} P_{\mathcal X}\pi^* = p, \\ P_{\mathcal Y}\pi^* = q,\end{aligned}
	$$
	and since $\pi^* \in \mathcal E_0$, we have $\pi^* \in \Pi_{\mathcal E_0}(p,q)$.
	Clearly, the norm $\pi \mapsto \|L^{1/2}\pi\|_{L^2}^2$ is weakly lower semicontinuous; hence
	$$
	\mathcal H(\pi^*) \le \liminf_{n\to\infty} \mathcal H(\pi_n).
	$$
	
	By Lemma \ref{lem: RestritionFunctionIsConti}, $T_c|_{\mathcal E_0}(\pi) = \langle c, \pi \rangle_{\mathcal E_0}$ is a continuous linear functional on $\mathcal E_0$. The Riesz representation theorem guarantees that there exists a unique $\eta_c \in \mathcal E_0$ such that for all $\pi \in \mathcal E_0$,
	$$
	T_c|_{\mathcal E_0}(\pi) = \langle c, \pi \rangle_{L^2} = \langle \pi, \eta_c \rangle_{\mathcal E_0}.
	$$
	Substituting $\pi = \pi_n$ and $\pi = \pi^*$ yields
	$$
	\begin{aligned}\langle c, \pi_n \rangle_{L^2} &= \langle \pi_n, \eta_c \rangle_{\mathcal E_0}, \\ \langle c, \pi^* \rangle_{L^2} &= \langle \pi^*, \eta_c \rangle_{\mathcal E_0}. \end{aligned}
	$$
	Since $\eta_c \in \mathcal E_0$, taking $\phi = \eta_c$ in the definition of weak convergence gives $\langle \pi_n, \eta_c \rangle_{\mathcal E_0} \to \langle \pi^*, \eta_c \rangle_{\mathcal E_0}$, and therefore
	$$
	\langle c, \pi_n \rangle_{L^2} \to \langle c, \pi^* \rangle_{L^2}.
	$$
	In summary,
	$$
	J_\varepsilon(\pi^*) = \langle c, \pi^* \rangle + \varepsilon \mathcal H(\pi^*)\le \liminf_{n\to\infty} J_\varepsilon(\pi_n) = a.
	$$
	Since $a$ is the infimum, we have $J_\varepsilon(\pi^*) = a$. Thus the existence of $\pi_\varepsilon^*$ is established.
	
	We now prove uniqueness. Since the linear term $\langle c, \cdot \rangle$ does not destroy convexity, and by Lemma \ref{lem: ConvexityOfDirletEnergy}, $\mathcal H(\pi)$ is strictly convex, the functional $J_\varepsilon(\pi)$ is strictly convex (for $\varepsilon > 0$). A strictly convex functional on a closed convex set $\Pi_{\mathcal E_0}(p,q)$ has at most one minimizer $\pi_\varepsilon^*$ (Lemma \ref{lem: StructrueOfFeasibleSol}). Since existence has already been established, the minimizer is unique.
\end{proof}

\begin{remark}
	The proof of Proposition \ref{prop: optimalEstimationExistUniq} contains the fact that the functional $J_\varepsilon(\pi)$ is bounded below. This follows from
	$$
	J_\varepsilon(\pi) = \langle c, \pi \rangle + \varepsilon \mathcal H(\pi)
	\ge -\|c\|_{L^2} \|\pi\|_{L^2} + \varepsilon \mathcal H(\pi).
	$$
	Using $\|\pi\|_{L^2} \le \sqrt{\frac{2}{\lambda}} \sqrt{\mathcal H(\pi)}$ to control $\|\pi\|_{L^2}$, we obtain
	$$
	J_\varepsilon(\pi) \ge - \|c\|_{L^2} \sqrt{\frac{2}{\lambda}} \, \sqrt{\mathcal H(\pi)} + \varepsilon \mathcal H(\pi).
	$$
	The minimum of the right-hand side is attained at $\sqrt{\mathcal H(\pi)} = \frac{\|c\|_{L^2}}{\varepsilon} \sqrt{\frac{1}{2\lambda}}$, so the lower bound of $J_\varepsilon(\pi)$ is given by
	$$
	\inf_{\pi \in \Pi_{\mathcal E_0}(p,q)} J_\varepsilon(\pi) \ge - \frac{\|c\|_{L^2}^2}{2\varepsilon \lambda} > -\infty.
	$$
\end{remark}

\section{Discussion}\label{sect: discussion}
Building upon Section \ref{sect: signedOT}, we construct and establish the existence of the global Euler–Lagrange equation, which in turn enables the design of an algorithm for finding the optimal solution. Moreover, although the objects of the variational problem in infinite-dimensional spaces are subject to more stringent conditions, more relaxed treatments are available, for instance, in the setting of finite measure spaces. Nevertheless, for convenience, certain propositions will directly require $\mathcal X$ and $\mathcal Y$ to be compact.
\subsection{Global Weak and Strong Forms of the E-L Equation}
Define the marginal operator $A: \mathcal E_0 \to \mathcal M_0(\mathcal X)\times\mathcal M_0(\mathcal Y)$, given explicitly by $A\pi=(P_{\mathcal X}\pi,P_{\mathcal Y}\pi)$. The variational problem is as follows:
$$\min_{\pi \in \mathcal E_0} J_\varepsilon(\pi), \quad \text{s.t. } A\pi = b,$$
where $b = (p,q)$. When the original optimization problem attains its minimum, the resulting equation is the weak Euler–Lagrange equation.

We first prove two important lemmas: Lemma \ref{lem: DenseLemma} and Lemma \ref{lem: Closed Image}, to understand the structure of the marginal projection operators and their adjoints (Remark \ref{remark: AdjointOperator} and Lemma \ref{lem: Adjoint Operator}), as well as the relationship between the finite-energy zero-mass space $\mathcal E_0$ and the infinite-energy state space $\mathcal M_0$, i.e., $e$. These results will, in turn, reveal the essential structure of the weak and strong Euler–Lagrange equations (Proposition \ref{prop: Weak-EL-Equation} and Proposition \ref{prop: Strong-EL-Equation}), and the information transport mechanism inherent in the regularized formulation (Remark \ref{remark: AdjointOperatorInformationTrans}).

\begin{Lem}[Closed Marginal Operator]\label{lem: Closed Image}
	The operator $A$ is a bounded linear operator. If $\mathcal X$ and $\mathcal Y$ are compact, then moreover $\operatorname{Im} A=\overline{\operatorname{Im} A}$, i.e., $\operatorname{Im} A$ is closed.
\end{Lem}
\begin{proof}
	It is easy to check the linearity of $A$ is inherited from the projection operators $P_{\mathcal X}$ and $P_{\mathcal Y}$. We now prove that $A$ is bounded. We have $A\pi\in\mathcal M_0(\mathcal X) \times \mathcal M_0(\mathcal Y) \subset L^2(\mathcal X) \times L^2(\mathcal Y)$, with norm defined by $\|A\pi\|_Y = \left( \|P_{\mathcal X}\pi\|_{L^2(\mathcal X)}^2 + \|P_{\mathcal Y}\pi\|_{L^2(\mathcal Y)}^2 \right)^{1/2}$.
	From the proof of Lemma \ref{lem: StructrueOfFeasibleSol}, we have $\|P_{\mathcal X}\pi\|_{L^2(\mathcal X)}^2 \le \nu(\mathcal Y) \|\pi\|_{\mathcal E}^2$ and $\|P_{\mathcal Y}\pi\|_{L^2(\mathcal Y)}^2 \leq \mu(\mathcal X) \|\pi\|_{\mathcal E_0}^2$. Combining these yields
	$$
	\begin{aligned}
		\|A\pi\|_Y^2
		&= \|P_{\mathcal X}\pi\|_{L^2(\mathcal X)}^2 + \|P_{\mathcal Y}\pi\|_{L^2(\mathcal Y)}^2 \\
		&\leq  (\nu(\mathcal Y) + \mu(\mathcal X)) \|\pi\|_{\mathcal E_0}^2.
	\end{aligned}
	$$
	Taking $C = \sqrt{\nu(\mathcal Y) + \mu(\mathcal X)}$, we obtain $\|A\pi\|_Y \leq C \|\pi\|_{\mathcal E_0}$ for all $\pi \in \mathcal E_0$, i.e., $A$ is bounded.
	
	We now prove that $\operatorname{Im} A$ is closed. We claim that $\operatorname{Im} A = \mathcal E_{\mathcal X}^0 \times \mathcal E_{\mathcal Y}^0$, where $\mathcal E_\mathcal X^0=D(L_\mathcal X^{1/2})\cap\mathcal M_0(\mathcal X)$:
	$$
	\mathcal E_\mathcal X^0=\left\{p\in L^2(\mathcal X):\hat p(0)=0,\int m_\mathcal X(\xi)|\hat p(\xi)|^2d\hat\mu<\infty\right\}. 
	$$
	For the first direction $\operatorname{Im} A \subset \mathcal E_{\mathcal X}^0 \times \mathcal E_{\mathcal Y}^0$, we prove that if $\pi \in \mathcal E_0$, then $(P_\mathcal X\pi,P_\mathcal Y\pi)\in\mathcal E_\mathcal X^0\times\mathcal E_\mathcal Y^0$.
	We first show that $A$ preserves zero mass, i.e., $P_{\mathcal X}\pi \in \mathcal M_0(\mathcal X)$. Integrating $P_{\mathcal X}\pi(x)$ yields
	$$
	\int_{\mathcal{X}} P_\mathcal X\pi(x)\,d\mu(x)=\int_\mathcal X\int_\mathcal Y\pi(x,y)d\nu(y)d\mu(x)=0,\quad \pi\in\mathcal E_0,
	$$
	so $P_\mathcal X\pi\in M_0(\mathcal X)$.
	
	Next, we prove the energy regularity of $P_{\mathcal X}\pi$, i.e., $P_\mathcal X\pi\in D(L_\mathcal X^{1/2})$. Let $p = P_{\mathcal X}\pi$. We have already shown that $P_{\mathcal X}\pi \in \mathcal M_0(\mathcal X) \subset L^2(\mathcal X)$; it remains to show $\int_{\widehat{\mathcal X}} (1 + m_{\mathcal X}(\xi)) |\hat{p}(\xi)|^2 \, d\hat{\mu}(\xi) < \infty$.
	By Remark \ref{remark: fixedMarginalDistrib}, we have $\hat{p}(\xi) = \hat{\pi}(\xi, 0)$, so it suffices to prove
	$$
	\int_{\widehat{\mathcal X}} (1 + m_{\mathcal X}(\xi)) |\hat{\pi}(\xi, 0)|^2 \, d\hat{\mu}(\xi) < \infty
	$$
	Since $\pi \in \mathcal E_0$, we have $\|\pi\|_{\mathcal E} < \infty$, i.e.,
	$$
	\int_{\widehat{\mathcal X}\times\widehat{\mathcal Y}} (1 + m_{\mathcal X}(\xi) + m_{\mathcal Y}(\eta)) |\hat{\pi}(\xi,\eta)|^2 \, d\hat{\mu}(\xi)d\hat{\nu}(\eta) < \infty.
	$$
	Since $\mathcal X$ and $\mathcal Y$ are locally compact abelian groups, the spectra on $\widehat{\mathcal X}$ and $\widehat{\mathcal Y}$ are given by discrete counting measures. $\eta = 0$ is a frequency slice, and its contribution is a subset of the full joint integral; hence
	$$
	\int_{\widehat{\mathcal X}} (1 + m_{\mathcal X}(\xi)) |\hat{\pi}(\xi, 0)|^2 \, d\hat{\mu}(\xi) \le \int_{\widehat{\mathcal X}\times\widehat{\mathcal Y}} (1 + m_{\mathcal X}(\xi) + m_{\mathcal Y}(\eta)) |\hat{\pi}(\xi,\eta)|^2 \, d\hat{\mu}(\xi)d\hat{\nu}(\eta) < \infty.
	$$
	Therefore $\|p\|_{\mathcal E_{\mathcal X}} < \infty$, i.e., $p \in D(L_{\mathcal X}^{1/2})$.
	In summary, we have $P_\mathcal X\pi\in\mathcal E_\mathcal X^0$, and similarly $P_\mathcal Y\pi\in\mathcal E_\mathcal Y^0$. Hence $\operatorname{Im} A \subset \mathcal E_{\mathcal X}^0 \times \mathcal E_{\mathcal Y}^0$.
	
	For the second direction $\operatorname{Im} A \supset \mathcal E_{\mathcal X}^0 \times \mathcal E_{\mathcal Y}^0$, take any $(p,q)\in\mathcal E_{\mathcal X}^0 \times \mathcal E_{\mathcal Y}^0$. We need to show that there exists $\pi \in \mathcal E_0$ such that $A\pi=(p,q)$. We claim that $\pi$ of the following form satisfies the requirements:
	$$
	\pi(x,y) := \frac{1}{\nu(\mathcal Y)} p(x) + \frac{1}{\mu(\mathcal X)} q(y).
	$$
	Verifying the marginals:
	$$
	\begin{aligned} \int_{\mathcal Y} \frac{p(x)}{\nu(\mathcal Y)}\,d\nu(y)
		= \frac{p(x)}{\nu(\mathcal Y)} \cdot \nu(\mathcal Y) = p(x),\\\int_{\mathcal X} \frac{q(y)}{\mu(\mathcal X)}\,d\mu(x)
		= \frac{q(y)}{\mu(\mathcal X)} \cdot \mu(\mathcal X) = q(y).\end{aligned}
	$$
	Since $(p,q) \in \mathcal E_{\mathcal X}^0 \times \mathcal E_{\mathcal Y}^0$, we have $\pi \in \mathcal M_0(\mathcal X \times \mathcal Y)$, i.e., the zero-mass constraint is preserved.
	
	It remains to show that $\pi$ has finite energy, i.e., $\pi \in \mathcal E$. Taking the Fourier transform of $\pi(x,y)$, we obtain
	$$
	\hat{\pi}(\xi,\eta) = \frac{1}{\nu(\mathcal Y)} \widehat{p(x)}(\xi,\eta) + \frac{1}{\mu(\mathcal X)} \widehat{q(y)}(\xi,\eta).
	$$
	By the orthogonality of characters on locally compact groups,
	$$
	\hat{\pi}(\xi,\eta) = \frac{1}{\nu(\mathcal Y)} \hat{p}(\xi)\delta_0(\eta) + \frac{1}{\mu(\mathcal X)} \delta_0(\xi)\hat{q}(\eta).
	$$
	The joint energy norm is given by
	$$
	\|\pi\|_{\mathcal E}^2 := \int_{\widehat{\mathcal X}\times\widehat{\mathcal Y}} (1 + m_{\mathcal X}(\xi) + m_{\mathcal Y}(\eta)) |\hat{\pi}(\xi,\eta)|^2 \, d\hat{\mu}(\xi)d\hat{\nu}(\eta).
	$$
	By Remark \ref{remark: fixedMarginalDistrib}, the marginal constraints occupy only the zero-frequency axes in the joint frequency space. Therefore,
	$$
	\operatorname{supp}\hat\pi
	\subset
	(\widehat{\mathcal X}\times{0})
	\cup
	({0}\times\widehat{\mathcal Y}).
	$$
	When $\eta = 0$, only $p(x)$ contributes; denote this term by $I_p$. When $\xi = 0$, only $q(y)$ contributes; denote this term by $I_q$. The cross term involving $\delta_0(\xi)\delta_0(\eta)$ contributes at $(\xi,\eta) = (0,0)$; denote this term by $I_{\text{cross}}$.
	
	By definition,
	$$
	\begin{aligned}
		I_p &:= \int_{\widehat{\mathcal X}\times\widehat{\mathcal Y}} (1 + m_{\mathcal X}(\xi) + m_{\mathcal Y}(\eta)) \left| \frac{1}{\nu(\mathcal Y)} \hat{p}(\xi) \delta_0(\eta) \right|^2 d\hat{\mu}(\xi)d\hat{\nu}(\eta) \\
		&= \frac{1}{\nu(\mathcal Y)^2} \int_{\widehat{\mathcal X}} (1 + m_{\mathcal X}(\xi) + m_{\mathcal Y}(0)) |\hat{p}(\xi)|^2 \, d\hat{\mu}(\xi).
	\end{aligned}
	$$
	By the Markov property of the generator, $m_{\mathcal Y}(0) = 0$. Hence,
	$$
	I_p = \frac{1}{\nu(\mathcal Y)^2} \int_{\widehat{\mathcal X}} (1 + m_{\mathcal X}(\xi)) |\hat{p}(\xi)|^2 \, d\hat{\mu}(\xi).
	$$
	where $p \in \mathcal E_{\mathcal X}^0$:
	$$
	\int_{\widehat{\mathcal X}} (1 + m_{\mathcal X}(\xi)) |\hat{p}(\xi)|^2 \, d\hat{\mu}(\xi) = \|p\|_{\mathcal E_{\mathcal X}}^2 < \infty.
	$$
	Thus $I_p < \infty$. Similarly, $I_q < \infty$.
	
	For the cross term, since $\pi \in \mathcal M_0(\mathcal X \times \mathcal Y)$, we have $\hat\pi(0,0) = 0$, so
	$$
	I_{\text{cross}} = \int_{\widehat{\mathcal X}\times\widehat{\mathcal Y}} (1 + m_{\mathcal X}(\xi) + m_{\mathcal Y}(\eta)) |\hat{\pi}(0,0)|^2 \, d\hat{\mu}(\xi)d\hat{\nu}(\eta)=0
	$$
	In summary,
	$$
	\|\pi\|_{\mathcal E_{\mathcal Z}}^2 = I_p + I_q + I_{\text{cross}} = I_p + I_q < \infty.
	$$
	Therefore $\pi \in \mathcal E$. Together with $\pi \in \mathcal M_0(\mathcal X \times \mathcal Y)$, this gives $\pi \in \mathcal E_0$. Hence $\operatorname{Im} A \supset \mathcal E_{\mathcal X}^0 \times \mathcal E_{\mathcal Y}^0$. Combining this with $\operatorname{Im} A \subset \mathcal E_{\mathcal X}^0 \times \mathcal E_{\mathcal Y}^0$, we obtain
	$$
	\operatorname{Im} A = \mathcal E_{\mathcal X}^0 \times \mathcal E_{\mathcal Y}^0
	$$
	Since $\mathcal E_{\mathcal X}^0$ and $\mathcal E_{\mathcal Y}^0$ are closed subspaces of $L^2(\mathcal X \times \mathcal Y)$, their product $\mathcal E_{\mathcal X}^0 \times \mathcal E_{\mathcal Y}^0$ is also closed. Therefore, $\operatorname{Im} A$ is closed.
\end{proof}

\begin{remark} \label{remark: exactDomain}
	The proof of Lemma \ref{lem: Closed Image} contains $\operatorname{Im} A = \mathcal E_{\mathcal X}^0 \times \mathcal E_{\mathcal Y}^0$, where $\mathcal E_\mathcal X^0=D(L_\mathcal X^{1/2})\cap\mathcal M_0(\mathcal X)$:
	$$
	\mathcal E_\mathcal X^0=\left\{p\in L^2(\mathcal X):\hat p(0)=0,\int m_\mathcal X(\xi)|\hat p(\xi)|^2d\hat\mu<\infty\right\}.
	$$
	This provides a more precise characterization of the domain $(\mathcal E_\mathcal X^0\times\mathcal E_\mathcal Y^0)^*$ of the adjoint operator $A^*$.
\end{remark}

\begin{remark} \label{remark: AdjointOperator}
	For the marginal operator $A: \mathcal E_0 \to \mathcal E_{\mathcal X}^0 \times \mathcal E_{\mathcal Y}^0$, its adjoint is given by $A^*:(\mathcal E_\mathcal X^0\times\mathcal E_\mathcal Y^0)^*\rightarrow \mathcal E_0^*$, the diagram is 
	\[
	\begin{CD}
		\displaystyle \mathcal{E}_0
		@>{\displaystyle A}>>
		\displaystyle \mathcal E_{\mathcal X}^0 \times \mathcal E_{\mathcal Y}^0
		\\
		@V{\displaystyle \cong}VV
		@
		VV{\displaystyle \cong}V
		\\
		\displaystyle \mathcal{E}_0^*
		@<{\displaystyle A^*}<<
		\displaystyle
		(\mathcal E_{\mathcal X}^0)^*
		\times
		(\mathcal E_{\mathcal Y}^0)^*
	\end{CD}
	\]
	Let $(\mathcal E_X^0)^* = \mathcal B(\mathcal E_X^0, \mathbb R)$ denote the space of all bounded linear functionals on $\mathcal E_X^0$. The completeness of $(\mathcal E_X^0)^*$ is inherited from $\mathbb R$. By the duality isomorphism for product spaces $(\mathcal E_\mathcal X^0\times\mathcal E_\mathcal Y^0)^* \cong (\mathcal E_\mathcal X^0)^*\times(\mathcal E_\mathcal Y^0)^*$, for any $F \in (\mathcal E_{\mathcal X}^0 \times \mathcal E_{\mathcal Y}^0)^*$, there exists a unique pairing $(\phi^*, \psi^*) \in (\mathcal E_\mathcal X^0)^*\times(\mathcal E_\mathcal Y^0)^*$ that represents it as
	$$
	F = (\phi^*, \psi^*).
	$$
	The construction of the isomorphism ensures that for all $(u, v) \in \mathcal E_\mathcal X^0\times\mathcal E_\mathcal Y^0$, $F(u, v) = \phi^*(u) + \psi^*(v)$. Although this is formally consistent with Lemma \ref{lem: Adjoint Operator}, the informational content of the two mappings is entirely different.
	
	\begin{center}
		\hspace*{-2cm}
		\begin{tikzcd}
			[row sep=4em,
			column sep=2em,
			nodes={inner sep=1pt, font=\small},
			every label/.append style={font=\small}]
			&
			(\mathcal E_{\mathcal X}^0)^*
			\times
			(\mathcal E_{\mathcal Y}^0)^*
			\arrow[rr, "{(\phi^*,\psi^*)\mapsto \phi^*+\psi^*}"]
			\arrow[rd, "{\cong}"']
			&&
			\mathcal E_0^*
			\\
			&
			&
			(\mathcal E_{\mathcal X}^0\times\mathcal E_{\mathcal Y}^0)^*
			\arrow[ru,"A^*"']
			&
		\end{tikzcd}
	\end{center}
	
	In Hilbert space, this means that for every $F \in (\mathcal E_{\mathcal X}^0 \times \mathcal E_{\mathcal Y}^0)^*$ and every $\pi \in \mathcal E_0$,
	$$
	\langle A^*F, \pi \rangle_{\mathcal E_0} = \langle F, A\pi \rangle_{\mathcal E_{\mathcal X}^0 \times \mathcal E_{\mathcal Y}^0}.
	$$
\end{remark}

\begin{Lem}[Duality]\label{lem: Adjoint Operator}
	$A^*(\phi^*, \psi^*) = \phi^* + \psi^*$, $\forall (\phi^*, \psi^*) \in (\mathcal E_\mathcal X^0)^*\times(\mathcal E_\mathcal Y^0)^*$.
\end{Lem}
\begin{proof}
	By Remark \ref{remark: AdjointOperator},
	$$
	\begin{aligned}
		\langle A^*(\phi^*, \psi^*), \pi \rangle_{\mathcal E_0}
		&= \langle (\phi^*, \psi^*), A\pi \rangle_{\mathcal E_{\mathcal X}^0 \times \mathcal E_{\mathcal Y}^0} \\
		&= \langle (\phi^*, \psi^*), (P_{\mathcal X}\pi, P_{\mathcal Y}\pi) \rangle_{\mathcal E_{\mathcal X}^0 \times \mathcal E_{\mathcal Y}^0} \\
		&= \langle\phi^*,P_{\mathcal X}\pi\rangle + \langle\psi^*,P_{\mathcal Y}\rangle \\
		&= \int_{\mathcal X} \phi^*(x) (P_{\mathcal X}\pi)(x) \, d\mu(x) + \int_{\mathcal Y} \psi^*(y) (P_{\mathcal Y}\pi)(y) \, d\nu(y) \\
		&= \int_{\mathcal X} \phi^*(x) \left( \int_{\mathcal Y} \pi(x,y) \, d\nu(y) \right) d\mu(x) + \int_{\mathcal Y} \psi^*(y) \left( \int_{\mathcal X} \pi(x,y) \, d\mu(x) \right) d\nu(y) \\
		&= \int_{\mathcal X \times \mathcal Y} \phi^*(x) \pi(x,y) \, d\mu(x)d\nu(y) + \int_{\mathcal X \times \mathcal Y} \psi^*(y) \pi(x,y) \, d\mu(x)d\nu(y) \\
		&= \int_{\mathcal X \times \mathcal Y} [\phi^*(x) + \psi^*(y)] \pi(x,y) \, d\mu(x)d\nu(y) \\
		&= \langle \phi^*(x) + \psi^*(y), \pi \rangle_{\mathcal E_0}.
	\end{aligned}
	$$
	Therefore,
	$$
	A^*(\phi^*, \psi^*) = \phi^* + \psi^*.
	$$
\end{proof}

\begin{remark} \label{remark: AdjointOperatorInformationTrans}
	The proof of Lemma \ref{lem: Adjoint Operator} contains 
	$$
	\begin{aligned}\langle\phi^*,P_\mathcal X\pi\rangle_\mathcal X=\langle \phi^*(x),\pi(x,y)\rangle_{\mathcal X\times \mathcal Y}\\ \langle\psi^*,P_\mathcal X\pi\rangle_\mathcal Y=\langle \psi^*(y),\pi(x,y)\rangle_{\mathcal X\times \mathcal Y}\end{aligned}
	$$
	This means that $P_{\mathcal X}^*\phi^*=\phi^*(x)$ and $P_{\mathcal Y}^*\psi^*=\psi^*(y)$, and consequently,
	$$
	\begin{aligned}A^*(\phi^*, \psi^*)= \phi^*(x) + \psi^*(y)&=P_{\mathcal X}^*\phi^* + P_{\mathcal Y}^*\psi^* \\ &=\begin{pmatrix}P_{\mathcal X}^*&P_{\mathcal Y}^*\end{pmatrix}  \begin{pmatrix}\phi^*\\ \psi^*  \end{pmatrix} \end{aligned}
	$$
	i.e., $A^*=\begin{pmatrix}P_{\mathcal X}^*&P_{\mathcal Y}^*\end{pmatrix}$. The lifting operator $P_{\mathcal X}^*$ lifts a function $\phi: \mathcal X \to \mathbb R$ to $P_X^*\phi:\mathcal X\times\mathcal Y\rightarrow\mathbb R$, copying the marginal potential function onto the joint space. A function $\phi(x)$ that originally depends only on $x$ becomes a function that depends on both $x$ and $y$, but is constant in the $y$-direction (i.e., does not vary with $y$). This implies that when integrating $\phi(x)$ against the marginal, no information is lost, it is merely multiplied by a constant factor $\nu(\mathcal Y)$:
	$$
	P_\mathcal X \phi(x,y) = \int_\mathcal Y \phi(x) \, d\nu(y) = \nu(\mathcal Y) \cdot \phi(x)
	$$
	Similarly, $\psi(y)$ is lifted to $(x,y)\mapsto\psi(y)$. Geometrically, the marginalization operators $P_{\mathcal X}$ and $P_{\mathcal Y}$ project a three-dimensional object onto two-dimensional planes, while the lifting operators $P_{\mathcal X}^*$ and $P_{\mathcal Y}^*$ pull a two-dimensional image back into a three-dimensional cylinder. If this cylinder is not twisted in the vertical direction and is vertically uniform, we say that $P_{\mathcal X}^*$ is a rigid embedding. Lemma \ref{lem: CoecivityOfLifitingOperator} further shows that the volume of this cylinder is at least the area of its base times a fixed height, it cannot degenerate to zero thickness.
\end{remark}

\begin{Prop}[Weak Euler–Lagrange Equation]\label{prop: Weak-EL-Equation}
	At the minimizer $\pi_\varepsilon^*$, the following hold:
	\begin{enumerate}
		\item[(a)] $D J_\varepsilon(\pi^*) \in (\ker A)^\perp$ and the local equation $\langle c, h \rangle + \varepsilon \langle L^{1/2}\pi_\varepsilon^*, L^{1/2}h \rangle = 0$ for $h \in \ker A$.
		\item[(b)] When $\mathcal X$ and $\mathcal Y$ are compact, the global equation $\langle c + \phi^*(x) + \psi^*(y), h \rangle + \varepsilon \langle L^{1/2}\pi_\varepsilon^*, L^{1/2}h \rangle = 0$ holds for all $h \in \mathcal E_0$, where $(\phi^*, \psi^*)\in(\mathcal E_X^0)^*\times(\mathcal E_Y^0)^*$.
	\end{enumerate}
\end{Prop}
\begin{proof}
	We first verify the first-order optimality condition. Let $\pi^*$ be the minimizer. For any direction $h \in \ker A$ (i.e., $P_{\mathcal X}h = 0$ and $P_{\mathcal Y}h = 0$), the curve $\pi^* + t h$ remains within the feasible set, since
	$$
	A(\pi^* + th) = A\pi^* + tAh = b + 0 = b.
	$$
	Since $\pi^*$ is a minimizer, $t = 0$ is a minimizer of the function $t \mapsto J_\varepsilon(\pi^* + th)$, so its first derivative vanishes:
	$$
	\frac{d}{dt} J_\varepsilon(\pi^* + th) = \langle c, h \rangle + \varepsilon \left(\langle L^{1/2}(\pi^* + th), L^{1/2}h \rangle+t \|L^{1/2}h\|^2\right).
	$$
	At $t = 0$, we have
	$$
	\langle c, h \rangle + \varepsilon \langle L^{1/2}\pi^*, L^{1/2}h \rangle = 0, \quad \forall h \in \ker A.
	$$
	This is the weak Euler–Lagrange equation on the tangent space $\ker A$. The gradient $D J_\varepsilon(\pi^*)$ of the objective functional vanishes in all tangent directions $h$ of the feasible set; that is, the gradient is orthogonal to the tangent space:
	$$
	D J_\varepsilon(\pi^*) \in (\ker A)^\perp.
	$$
	Its explicit form is $D J_\varepsilon(\pi^*)[h] = \langle c, h \rangle + \varepsilon \langle L^{1/2}\pi^*, L^{1/2}h \rangle = 0$ for all $h \in \ker A$.
	Equation (EL-weak) holds only on $\ker A$. To obtain the global equation (for arbitrary $h \in \mathcal E_0$), we introduce Lagrange multipliers $(\phi^*, \psi^*)$. Since $\operatorname{Im} A$ is closed by Lemma \ref{lem: Closed Image}, the closed range theorem yields
	$$
	(\ker A)^\perp = \operatorname{Im} A^*.
	$$
	$DJ_\varepsilon(\pi^*) \in \operatorname{Im} A^*$ implies that there exists $F \in (\mathcal E_{\mathcal X}^0 \times \mathcal E_{\mathcal Y}^0)^*$ such that
	$$
	DJ_\varepsilon(\pi^*) = A^*F.
	$$
	By Remark \ref{remark: AdjointOperator}, for any $F \in (\mathcal E_{\mathcal X}^0 \times \mathcal E_{\mathcal Y}^0)^*$, there exists a unique pairing $(-\phi^*, -\psi^*) \in (\mathcal E_\mathcal X^0)^*\times(\mathcal E_\mathcal Y^0)^*$ representing it as $F = -(\phi^*, \psi^*)$. Therefore,
	$$
	DJ_\varepsilon(\pi^*) = -A^*(\phi^*, \psi^*).
	$$
	By Lemma \ref{lem: Adjoint Operator}, both $A^*(\phi^*, \psi^*): \mathcal E_0 \to \mathbb R$ and the variation $DJ_\varepsilon(\pi^*): \mathcal E_0 \to \mathbb R$ are elements of $\mathcal E_0^*$. Applying both functionals to the same direction $h \in \mathcal E_0$ and using Lemma \ref{lem: Adjoint Operator}, we obtain
	$$
	D J_\varepsilon(\pi^*)[h] = - \langle \phi^* + \psi^*, h \rangle_{\mathcal E_0}, \quad \forall h \in \mathcal E_0.
	$$
	Expanding this yields the global weak Euler–Lagrange equation:
	$$
	\langle c + \phi^*(x) + \psi^*(y), h \rangle + \varepsilon \langle L^{1/2}\pi^*, L^{1/2}h \rangle = 0, \quad \forall h \in \mathcal E_0.
	$$
\end{proof}

Lemma \ref{lem: DenseLemma} states that $\mathcal M_0$ can be approximated by $\mathcal E_0$ via spectral truncation and zero-frequency projection, meaning that every zero-mass $L^2$ state (possibly with infinite energy) can be approximated by finite-energy zero-mass states. We will rely on this lemma to bridge the weak and strong Euler–Lagrange equations.

\begin{Lem}[Density Lemma]\label{lem: DenseLemma}
	When $\mathcal X$ and $\mathcal Y$ are compact, $\mathcal E_0 \stackrel{dense}{\subset} \mathcal M_0$ holds in the $L^2$-sense, i.e., $\overline{\mathcal E_0}^{L^2}=\mathcal M_0$.
\end{Lem}
\begin{proof}
	By Remark \ref{remark: halfdesnseRemark}, it suffices to prove $\mathcal M_0 \subset \overline{\mathcal E_0}^{L^2}$. This is equivalent to showing that for every $f \in \mathcal M_0$, there exists a sequence $\{\pi_n\} \subset \mathcal E_0$ such that $\pi_n \to f$ in $L^2$. We will explicitly construct such a sequence $\{\pi_n\}$.
	We first prove $\overline{D(L^{1/2})}^{L^2}_{\mathcal X\times\mathcal Y}=L^2$. Since $L_{\mathcal X\times\mathcal Y}$ is a nonnegative self-adjoint operator, the spectral theorem gives $L_{\mathcal X\times\mathcal Y}=\int_0^\infty\lambda dE_\lambda$. Define the spectral truncation $P_n=E([0,n])$. For any $f \in L^2(\mathcal X \times \mathcal Y)$, let $g_n = P_n f$. By $\|E([0,n])f\|^2\leq\|E([0,\infty))f\|^2$,
	$$
	\|Lf_n\|^2=\int_0^n\lambda^2d\|E_\lambda f\|^2 \leq n^2\int_0^n d\|E_\lambda f\|^2 \leq n^2\|f\|^2<\infty,
	$$
	i.e., $L g_n = \int_0^n \lambda \, dE_\lambda g \in L^2$, so $g_n \in D(L_{\mathcal X\times\mathcal Y}) \subset D(L^{1/2}_{\mathcal X\times\mathcal Y})$.
	From $\mu_f([0,\infty))=\|E([0,\infty))f\|^2=\|f\|^2<\infty$, the tail of the spectral measure converges:
	$$
	\|f - g_n\|_{L^2}^2 = \int_{(n, \infty)} d\|E_\lambda f\|^2 \to 0 \quad (n \to \infty).
	$$
	Therefore, $\overline{D(L^{1/2}_{\mathcal X\times\mathcal Y})}^{L^2} = L^2(\mathcal X\times\mathcal Y)$.
	
	We now eliminate the zero-frequency component of $g_n$. Define the zero-mean orthogonal projection
	$$
	P_0: L^2(\mathcal X\times\mathcal Y) \to \mathcal M_0, \quad P_0 g = g - \frac{1}{\lambda(\mathcal X\times\mathcal Y)} \left( \int_{\mathcal X\times\mathcal Y} g \, d\lambda \right) \mathbf{1},
	$$
	where $\lambda(\mathcal X\times\mathcal Y) = \|\mathbf{1}\|_{L^2}^2$. Rearranging yields
	$$
	g-P_0 g =  \frac{1}{\lambda(\mathcal X\times\mathcal Y)} \left( \int_{\mathcal X\times\mathcal Y} g \, d\lambda \right) \mathbf{1}\in\text{span}\{1\}.
	$$
	From the above and Remark \ref{remark: geometricIntuition}, we have $g - P_0 g \perp \mathcal M_0$, i.e., $P_0$ is the orthogonal projection from $L^2$ onto its closed subspace $\mathcal M_0$. It removes the component along the constant direction $\mathbf{1}$. It is easily verified that in the Fourier domain,
	$$
	\widehat{P_0 g}(\xi) =\begin{cases}0, & \xi = 0, \\\hat g(\xi), & \xi \neq 0,\end{cases}
	$$
	i.e., it removes the zero frequency.
	By Remark \ref{remark: FourierCoefficient}, we further have $L^{1/2}\mathbf{1} = 0$, i.e., $\mathbf{1} \in D(L^{1/2}_{\mathcal X\times\mathcal Y})$. Since $g_n \in D(L^{1/2}_{\mathcal X\times\mathcal Y})$ and $D(L^{1/2})$ is a linear space, for every $g_n \in D(L^{1/2}_{\mathcal X\times\mathcal Y})$,
	$$
	\pi_n := P_0 g_n = g_n - a_n \mathbf{1} \in D(L^{1/2}).
	$$
	Moreover, this also shows that the projection operator preserves the energy: $\|L^{1/2}\pi_n\|_{L^2} = \|L^{1/2}g_n\|_{L^2}$.
	It is easy to check that $\pi_n \in \mathcal M_0$:
	$$
	\int_{\mathcal X\times\mathcal Y} \pi_n \, d\lambda
	= \int g_n - a_n \int \mathbf{1}
	= \int g_n - \frac{1}{\lambda(\mathcal X\times\mathcal Y)} \left( \int g_n \right) \lambda(\mathcal X\times\mathcal Y)
	= 0.
	$$
	Therefore, $\pi_n \in \mathcal M_0 \cap D(L^{1/2}) = \mathcal E_0$.
	Finally, we prove that $\pi_n \to f$. Since $f \in \mathcal M_0$, we have $P_0 f = f$. Moreover, by the properties of orthogonal projections, $\|P_0 g\|^2 \le \|g\|^2$, i.e., $\|P_0\| := \sup_{\|g\|=1} \|P_0 g\| \le 1$. Hence the projection operator is continuous. Since $g_n \to f$, we have
	$$
	\pi_n = P_0 g_n \to P_0 f = f.
	$$
	In summary, $\overline{\mathcal E_0}^{L^2}=\mathcal M_0$.
\end{proof}

\begin{remark}\label{remark: halfdesnseRemark}
	$\mathcal E_0 \stackrel{dense}{\nsubseteq} L^2(\mathcal X\times\mathcal Y)$, since $\mathcal E_0 \subset \mathcal M_0 = \{\text{span}\{1\}\}^{\perp}$, while $\{\text{span}\{1\}\}^{\perp}$ is a closed subspace, the closure of $\mathcal E_0$ is still contained in $\{\text{span}\{1\}\}^{\perp}$, i.e., $\overline{\mathcal E_0}^{L^2} \subset \mathcal M_0$. Therefore, $\mathcal E_0$ cannot be dense in the whole space $L^2(\mathcal X\times\mathcal Y)$. 
\end{remark}

\begin{Prop}[Strong Euler–Lagrange Equation]\label{prop: Strong-EL-Equation}
	When $\mathcal X$ and $\mathcal Y$ are compact, and $\pi^*$ is $D(L_{\mathcal X\times\mathcal Y})-$regular, i.e., $\pi^* \in D(L_{\mathcal X\times\mathcal Y})$, then we have the global strong Euler–Lagrange equation holds on $L^2(\mathcal X\times\mathcal Y)$:
	$$
	c + \phi^*(x) + \psi^*(y) + \varepsilon L\pi^* = C,\quad C\in\mathbb{R}
	$$
\end{Prop}
\begin{proof}
	The regularity $\pi^* \in D(L_{\mathcal X\times\mathcal Y})$ implies $L_{\mathcal X\times\mathcal Y}\pi^* \in L^2(\mathcal X\times\mathcal Y)$. By the self-adjointness of $L^{1/2}_{\mathcal X\times\mathcal Y}$,
	$$
	\langle L^{1/2}\pi^*, L^{1/2}h \rangle_{L^2}
	= \langle L\pi^*, h \rangle_{L^2},
	\quad \forall h \in \mathcal E_0.
	$$
	Substituting this into the global weak Euler–Lagrange equation yields $\langle c + \phi^* + \psi^*, h \rangle_{L^2}+ \varepsilon \langle L\pi^*, h \rangle_{L^2}= 0$ for all $h \in \mathcal E_0$, i.e.,
	$$
	\langle c + \phi^* + \psi^* + \varepsilon L\pi^*, h \rangle_{L^2}
	= 0,
	\quad \forall h \in \mathcal E_0.
	$$
	The above equation means that $f := c + \phi^* + \psi^* + \varepsilon L\pi^*$ is orthogonal to every element of $\mathcal E_0$. By the density lemma (Lemma \ref{lem: DenseLemma}) and the continuity of the inner product, $f$ is orthogonal to every element of $\mathcal M_0$. Hence $f \in (\mathcal M_0)^\perp$. By Remark \ref{remark: geometricIntuition}, we have $f\in \operatorname{span}\{\mathbf{1}_{\mathcal X\times\mathcal Y}\}$, so there exists a constant $C \in \mathbb R$ such that $f = C \cdot \mathbf{1}_{\mathcal X\times\mathcal Y}$. That is, in $L^2(\mathcal X \times \mathcal Y)$,
	$$
	c + \phi^* + \psi^* + \varepsilon L\pi^* = C.
	$$
\end{proof}

\subsection{Duality and Embedding of Zero-Sum Game}
The content of this subsection does not require any additional regularity of $\pi$; it suffices that $\pi \in \mathcal E_0$ to define $\mathcal L(\pi, \phi, \psi)$ and to apply the global weak Euler–Lagrange equation (Proposition \ref{prop: Weak-EL-Equation}), provided that $\mathcal X$ and $\mathcal Y$ are compact. For the original optimization problem, we introduce Lagrange multipliers $(\phi, \psi) \in (\mathcal E_\mathcal X^0)^*\times(\mathcal E_\mathcal Y^0)^*$ and define the Lagrangian functional:
$$
\mathcal L(\pi, \phi, \psi):= J_\varepsilon(\pi)+ \langle \phi, P_{\mathcal X}\pi - p \rangle_{L^2(\mathcal X)}+ \langle \psi, P_{\mathcal Y}\pi - q \rangle_{L^2(\mathcal Y)}.
$$
By Remark \ref{remark: AdjointOperatorInformationTrans} and Lemma \ref{lem: Adjoint Operator}, the above can be rewritten as
\begin{equation}\label{eq: LagrangianForm}
	\mathcal L(\pi, \phi, \psi)= J_\varepsilon(\pi)+ \langle \phi(x) + \psi(y), \pi \rangle_{L^2(\mathcal X \times \mathcal Y)}- \langle p, \phi \rangle_{L^2(\mathcal X)}- \langle q, \psi \rangle_{L^2(\mathcal Y)}.
\end{equation}

\begin{Prop}[Saddle Point]\label{prop: Saddle Point}
	The Lagrangian $\mathcal L$ admits a saddle point $(\pi^*, \phi^*, \psi^*)$ satisfying
	$$
	\mathcal L(\pi^*, \phi, \psi) \le \mathcal L(\pi^*, \phi^*, \psi^*) \le \mathcal L(\pi, \phi^*, \psi^*),\quad \forall \pi \in \mathcal E_0,\; \phi, \psi.
	$$
\end{Prop}
\begin{proof}
	We first prove the first inequality, i.e., the maximality in the dual direction. Fix $\pi^*$ as a feasible solution of the primal problem, i.e., for all $\pi^* \in \mathcal E_0$,
	$$
	\langle \phi, P_{\mathcal X}\pi^* - p \rangle = 0, \quad
	\langle \psi, P_{\mathcal Y}\pi^* - q \rangle = 0.
	$$
	Thus $\mathcal L(\pi^*, \phi, \psi) = J_\varepsilon(\pi^*)$ for all $\phi, \psi$. In particular, $\mathcal L(\pi^*, \phi^*, \psi^*) = J_\varepsilon(\pi^*)$. Therefore,
	$$
	\mathcal L(\pi^*, \phi, \psi) = \mathcal L(\pi^*, \phi^*, \psi^*), \quad \forall \phi, \psi.
	$$
	The first inequality $\mathcal L(\pi^*, \phi, \psi) \le \mathcal L(\pi^*, \phi^*, \psi^*)$ also holds.
	
	We now prove the minimality in the primal direction. Fix $(\phi^*, \psi^*)$, and define $F(\pi) := \mathcal L(\pi, \phi^*, \psi^*)$. By formula \ref{eq: LagrangianForm},
	$$
	F(\pi) = J_\varepsilon(\pi) + \langle \phi^*(x) + \psi^*(y), \pi \rangle - \langle p, \phi^* \rangle - \langle q, \psi^* \rangle.
	$$
	The Fréchet derivative of $F$ is given by
	$$
	D F(\pi)[h] = D J_\varepsilon(\pi)[h] + \langle \phi^*(x) + \psi^*(y), h \rangle.
	$$
	By the global weak Euler–Lagrange equation, the minimizer $\pi^*$ satisfies $D F(\pi^*) = 0$ in $\mathcal E_0^*$, i.e., $\pi^*$ is a stationary point of $F$. By Lemma \ref{lem: ConvexityOfDirletEnergy}, it is easy to see that $F(\pi) := \mathcal L(\pi, \phi^*, \psi^*)$ is strictly convex; hence $\pi^*$ is the unique global minimizer of $F(\pi)$. Therefore,
	$$
	\mathcal L(\pi^*, \phi^*, \psi^*) \le \mathcal L(\pi, \phi^*, \psi^*), \quad \forall \pi \in \mathcal E_0
	$$
	In summary,
	$$
	\mathcal L(\pi^*, \phi, \psi) \le \mathcal L(\pi^*, \phi^*, \psi^*) \le \mathcal L(\pi, \phi^*, \psi^*),\quad \forall \pi \in \mathcal E_0,\; \phi, \psi.
	$$
\end{proof}

\begin{Prop}[Strong Duality]\label{prop: StrongDuality}
	The saddle point $(\pi^*, \phi^*, \psi^*)$ is a solution to the following minimax problem:
	$$
	\begin{aligned} \inf_{\pi \in \mathcal E_0} \sup_{\phi \in (\mathcal E_\mathcal X^0)^*, \psi \in (\mathcal E_\mathcal Y^0)^*} \mathcal L(\pi, \phi, \psi),   \\   \sup_{\phi \in (\mathcal E_\mathcal X^0)^*, \psi \in (\mathcal E_\mathcal Y^0)^*} \inf_{\pi \in \mathcal E_0} \mathcal L(\pi, \phi, \psi),
	\end{aligned}
	$$
	and moreover,
	$$
	\inf_{\pi} \sup_{\phi,\psi} \mathcal L(\pi, \phi, \psi)=\sup_{\phi,\psi} \inf_{\pi} \mathcal L(\pi, \phi, \psi)=\mathcal L(\pi^*, \phi^*, \psi^*).
	$$
\end{Prop}
\begin{proof}
	We first prove $\inf_{\pi} \sup_{\phi,\psi} \mathcal L(\pi, \phi, \psi)=\mathcal L(\pi^*, \phi^*, \psi^*)$. Taking the minimizer $\pi^*$, by the same argument as in Proposition \ref{prop: Saddle Point}, we have $\mathcal L(\pi^*, \phi, \psi) = \mathcal L(\pi^*, \phi^*, \psi^*)$ for all $\phi, \psi$. Hence,
	$$
	\sup_{\phi,\psi} \mathcal L(\pi^*, \phi, \psi) = \mathcal L(\pi^*, \phi^*, \psi^*).
	$$
	Therefore,
	$$
	\inf_{\pi} \sup_{\phi,\psi} \mathcal L(\pi, \phi, \psi)
	\le \sup_{\phi,\psi} \mathcal L(\pi^*, \phi, \psi)
	= \mathcal L(\pi^*, \phi^*, \psi^*).
	$$
	On the other hand, from the right-hand inequality of the saddle point, $\mathcal L(\pi^*, \phi^*, \psi^*) \le \mathcal L(\pi, \phi^*, \psi^*)$ for all $\pi$, together with the pointwise inequality $\mathcal L(\pi,\phi^*,\psi^*)\leq \sup_{\phi,\psi}\mathcal L(\pi,\phi,\psi)$, we obtain
	$$
	\mathcal L(\pi^*, \phi^*, \psi^*) \le \inf_{\pi} \mathcal L(\pi, \phi^*, \psi^*)
	\le \inf_{\pi} \sup_{\phi,\psi} \mathcal L(\pi, \phi, \psi).
	$$
	Hence,
	$$
	\inf_{\pi} \sup_{\phi,\psi} \mathcal L(\pi, \phi, \psi) = \mathcal L(\pi^*, \phi^*, \psi^*).
	$$
	We now prove $\sup_{\phi,\psi} \inf_{\pi} \mathcal L(\pi, \phi, \psi) = \mathcal L(\pi^*, \phi^*, \psi^*)$. From the pointwise inequality $\inf_\pi\mathcal L(\pi,\phi,\psi)\le\mathcal L(\pi^*,\phi,\psi)$ and the left-hand inequality of the saddle point, $\mathcal L(\pi^*, \phi, \psi) \le \mathcal L(\pi^*, \phi^*, \psi^*)$ for all $\phi, \psi$, we have
	$$
	\inf_\pi\mathcal L(\pi,\phi,\psi)\le\mathcal L(\pi^*,\phi,\psi)\le\mathcal L(\pi^*,\phi^*,\psi^*),\quad \forall \phi, \psi.
	$$
	Therefore,
	$$
	\sup_{\phi,\psi} \inf_{\pi} \mathcal L(\pi, \phi, \psi)  \le  \mathcal L(\pi^*, \phi^*, \psi^*).
	$$
	On the other hand, from the pointwise inequality and the fact that $\pi^*$ is the global minimizer of $\mathcal L$, we obtain
	$$
	\sup_{\phi,\psi} \inf_{\pi} \mathcal L(\pi, \phi, \psi)\ge \inf_{\pi} \mathcal L(\pi, \phi^*, \psi^*)= \mathcal L(\pi^*, \phi^*, \psi^*).
	$$
	Hence,
	$$
	\sup_{\phi,\psi} \inf_{\pi} \mathcal L(\pi, \phi, \psi) = \mathcal L(\pi^*, \phi^*, \psi^*).
	$$
	In summary,
	$$
	\inf_{\pi} \sup_{\phi,\psi} \mathcal L(\pi, \phi, \psi)=\sup_{\phi,\psi} \inf_{\pi} \mathcal L(\pi, \phi, \psi)=\mathcal L(\pi^*, \phi^*, \psi^*).
	$$
	and $(\pi^*, \phi^*, \psi^*)$ is simultaneously a solution to both the $\inf\limits_{\pi \in \mathcal E_0} \sup\limits_{\phi \in (\mathcal E_\mathcal X^0)^*, \psi \in (\mathcal E_\mathcal Y^0)^*} \mathcal L(\pi, \phi, \psi)$ and $\sup\limits_{\phi \in (\mathcal E_\mathcal X^0)^*, \psi \in (\mathcal E_\mathcal Y^0)^*} \inf\limits_{\pi \in \mathcal E_0} \mathcal L(\pi, \phi, \psi)$ problems.
\end{proof}

In fact, the above variational procedure embeds a zero-sum game on a continuous strategy space. Let the payoffs of Player 1 and Player 2 be given respectively by
$$
\begin{aligned} U_1(\pi,\phi,\psi)&=-\mathcal L(\pi,\phi,\psi), \\U_2(\pi,\phi,\psi)&=\mathcal L(\pi,\phi,\psi). \end{aligned}
$$
Then $U_1 + U_2 = 0$. The primal player (Player 1) chooses a strategy $\pi \in \mathcal E_0$, with the objective of minimizing the transport cost $\mathcal L$. The dual player (Player 2) chooses a strategy $(\phi, \psi) \in (\mathcal E_\mathcal X^0)^*\times(\mathcal E_\mathcal Y^0)^*$, with the objective of maximizing the penalty $\mathcal L$. If the dual player moves first, the game takes the form
$$
\min_{\pi}\max_{\phi,\psi}\mathcal L(\pi,\phi,\psi).
$$
We interpret the saddle point from the perspective of best responses. Fixing $\pi^*$, the left inequality states that the dual player cannot obtain a higher payoff by changing $(\phi, \psi)$. Hence $(\phi^*, \psi^*)$ is a best response:
$$
\mathcal L(\pi^*, \phi, \psi) \le \mathcal L(\pi^*, \phi^*, \psi^*).
$$
Fixing $(\phi^*, \psi^*)$, the right inequality states that the primal player cannot reduce the cost by changing $\pi$. Hence $\pi^*$ is a best response:
$$
\mathcal L(\pi^*, \phi^*, \psi^*) \le \mathcal L(\pi, \phi^*, \psi^*).
$$
Thus the saddle point $(\pi^*, \phi^*, \psi^*)$ constitutes mutual best responses for both players. That is, the saddle point is a Nash equilibrium.

Proposition 2 shows that the order of moves does not affect the players' best responses, i.e., the Nash equilibrium of the game. If the primal player chooses $\pi$ first, and the dual player then observes the action and chooses the worst-case $(\phi, \psi)$:
$$
V^-=\inf_\pi\sup_{\phi,\psi}\mathcal L(\pi,\phi,\psi).
$$
If the dual player first sets the penalty, and the primal player then seeks the lowest-cost response:
$$
V^+=\sup_{\phi,\psi}\inf_\pi\mathcal L(\pi,\phi,\psi).
$$
In these cases, we have $V^-=V^+=\mathcal L(\pi^*, \phi^*, \psi^*)$.
\subsection{Decoupling Mapping and Spectrum}
Following the sign convention, we rewrite the strong Euler–Lagrange equation as $\varepsilon L\pi+c-\phi(x)-\psi(y)=0$. Define the Green operator $G := L_{\mathcal X\times\mathcal Y}^{-1}$, which is well-defined by Remark \ref{remark: Invertible}. Let $G_\varepsilon := \frac{1}{\varepsilon} G$; then
$$\pi = G_\varepsilon(\phi, \psi) := G_\varepsilon(\phi + \psi - c) = \frac1\varepsilon L_{\mathcal X\times\mathcal Y}^{-1}(\phi+\psi-c).$$
Setting $z = (\phi, \psi) \in (\mathcal E_{\mathcal X}^0)^* \times (\mathcal E_{\mathcal Y}^0)^*$, by Lemma \ref{lem: Adjoint Operator}, we obtain
$$\pi = G_\varepsilon(A^* z - c).$$
The self-consistency equation satisfied by $z = (\phi, \psi)$ is given by the marginal constraints $A G_\varepsilon (A^* z - c) = b$, i.e.,
$$A G A^* z = \varepsilon b + A G c.$$
Define the block operator $K := A G A^*$ and the constant $d := \varepsilon b + A G c$. The dual operator equation then takes the form
$$K z = d.$$
Rearranging yields $z=z-\alpha(Kz-d)$ (with $\alpha > 0$), so the solution to the above equation is precisely the fixed point $z^*$ of the decoupling operator $T(z):=z-\alpha(Kz-d)$.

\begin{Prop}[Richardson Iteration]\label{prop: RichardsonIteration}
	The decoupling equation $z^{n+1} = (I - \alpha^* K)z^n + \alpha^* d$ converges to the unique $z^*$, where $\alpha^* = \frac{2}{\lambda_{\max} + \lambda_{\min}}$ is the optimal step size:
	$$
	\begin{pmatrix}
		\phi^{n+1}\\
		\psi^{n+1}
	\end{pmatrix}
	=
	\begin{pmatrix}
		\phi^n\\
		\psi^n
	\end{pmatrix}
	-
	\alpha
	\begin{pmatrix}
		K_{\mathcal X\mathcal X}\phi^n+K_{\mathcal X\mathcal Y}\psi^n-d_\mathcal X\\
		K_{\mathcal Y\mathcal X}\phi^n+K_{\mathcal Y\mathcal Y}\psi^n-d_\mathcal Y
	\end{pmatrix}.
	$$
\end{Prop}
\begin{remark}\label{remark: OperatorK is Self-Ajoint}
	The proof or Proposition \ref{prop: RichardsonIteration} implies that $K$ is a bounded positive self-adjoint operator, i.e., $K=K^*$, and that $\langle z, K z \rangle > 0$ for all $\forall z \in (\mathcal E_{\mathcal X}^0)^* \times (\mathcal E_{\mathcal Y}^0)^*$.
\end{remark}

We seek further results on the state of information transport, for instance, whether information is lost or distorted during iteration, whether the original dimensional information is preserved, and whether the transport process itself is legitimate, so as to understand the mechanism underlying the reconstruction of the full joint state from marginal observations. These questions cannot be addressed by the baseline model above. Moreover, the baseline model optimizes with respect to the gradient; if the error could be precisely eliminated via preconditioning, the iteration would converge faster.

Expanding the block operator $K=AGA^*$, by Remark \ref{remark: AdjointOperatorInformationTrans}, we obtain
$$K = \begin{pmatrix}
	K_{XX} & K_{XY} \\
	K_{YX} & K_{YY}
\end{pmatrix}
=\begin{pmatrix}
	P_{\mathcal X}\\
	P_{\mathcal Y}
\end{pmatrix}
G \begin{pmatrix}P_{\mathcal X}^* & P_{\mathcal Y}^*\end{pmatrix}
= \begin{pmatrix}
	P_{\mathcal X}G P_{\mathcal X}^* & P_{\mathcal X}G P_{\mathcal Y}^* \\
	P_{\mathcal Y}G P_{\mathcal X}^* & P_{\mathcal Y}G P_{\mathcal Y}^* \end{pmatrix}
$$

\begin{remark}
	We take $K_{\mathcal X\mathcal Y}=P_{\mathcal X}G P_{\mathcal Y}^*=P_{\mathcal X} \circ G \circ P_{\mathcal Y}^*$ as an example to explain the block operator. A potential function $\psi \in (\mathcal E_{\mathcal Y}^0)^*$ defined on $\mathcal E_{\mathcal Y}^0$ is lifted to $\mathcal E_{\mathcal X}^0 \times \mathcal E_{\mathcal Y}^0$ via $P_{\mathcal Y}^*$, i.e., $(P_{\mathcal Y}^* \psi)(x,y) = \psi(y)$. The Green operator $G = L^{-1}$ propagates the source term $\psi(y)$ throughout $\mathcal E_{\mathcal X}^0 \times \mathcal E_{\mathcal Y}^0$, generating a response field $G(P_{\mathcal Y}^* \psi)$. Finally, $P_{\mathcal X}$ marginalizes the field on the product space back to $\mathcal X$, i.e., $P_{\mathcal X} \left( G(P_{\mathcal Y}^* \psi) \right)$, yielding a potential function $P_{\mathcal X} G P_{\mathcal Y}^* \psi\in (\mathcal E_{\mathcal X}^0)^*$ on the $\mathcal X$-space, since
	$$
	(P_{\mathcal X} G P_{\mathcal Y}^* \psi)(x) = \int_{\mathcal Y} (G(P_{\mathcal Y}^* \psi))(x,y) \, d\nu(y)
	$$
\end{remark}

Define the update rule for the decoupling equation $K z = d$ as
$$\begin{aligned}
	K_{XX} \phi^{n+1} &= d_X - K_{XY} \psi^n, \\ K_{YY} \psi^{n+1} &= d_Y - K_{YX} \phi^{n+1},
\end{aligned}$$
which is the block Gauss–Seidel iteration. One first updates $\phi$ using the most recent $\psi^n$, and then updates $\psi$ using the just-updated $\phi^{n+1}$. Let the unique solution of the decoupling iteration be $z^* = (\phi^*, \psi^*) \in (\mathcal E_{\mathcal X}^0)^* \times (\mathcal E_{\mathcal Y}^0)^*$. The algorithm flowchart is as follows. 

\begin{algorithm}[H]
	\caption{
		Spectral Sinkhorn-Type Block Gauss--Seidel Iteration
		for Regularized Signed OT
	}
	
	\begin{algorithmic}[1]
		
		\Require
		Initial potentials
		$z^0=(\phi^0,\psi^0)$,
		operator blocks
		$K_{\mathcal X\mathcal X},K_{\mathcal X\mathcal Y},K_{\mathcal Y\mathcal X},K_{\mathcal Y\mathcal Y}$,
		right hand side
		$d=(d_\mathcal X,d_\mathcal Y)$,
		tolerance $\delta$
		
		\Ensure
		Optimal potentials
		$(\phi^*,\psi^*)$
		and transport plan $\pi^*$

		\While{$\|z^{n+1}-z^n\|>\delta$}
		
		\State Compute current transport plan:
		
		\[
		\pi^n
		=
		\frac1{\varepsilon}
		L^{-1}
		(A^*z^n-c)
		\]
		
		\State Update $\mathcal X$-potential:
		
		\[
		\phi^{n+1}
		=
		K_{\mathcal X\mathcal X}^{-1}
		(d_\mathcal X-K_{\mathcal X\mathcal Y}\psi^n)
		\]

		\State Update $\mathcal Y$-potential:
		
		\[
		\psi^{n+1}
		=
		K_{\mathcal Y\mathcal Y}^{-1}
		(d_\mathcal Y-K_{\mathcal Y\mathcal X}\phi^{n+1})
		\]

		\State Assemble:
		
		\[
		z^{n+1}
		=
		(\phi^{n+1},\psi^{n+1})
		\]

		\EndWhile

		\State Recover optimal coupling:
		
		\[
		\pi^*
		=
		\frac1{\varepsilon}
		L^{-1}
		(A^*z^*-c)
		\]

	\end{algorithmic}
\end{algorithm}

We analyze the convergence conditions via the error propagation equation. Define the errors:
$$e_\phi^n := \phi^n - \phi^*, \quad e_\psi^n := \psi^n - \psi^*.$$
Then the error propagation for the $\phi$-equation and the $\psi$-equation are respectively given by
$$\begin{aligned}K_{XX} (\phi^{n+1} - \phi^*) &= -K_{XY} (\psi^n - \psi^*),\\ K_{YY} (\psi^{n+1} - \psi^*) &= -K_{YX} (\phi^{n+1} - \phi^*),\end{aligned}$$
i.e. (provided that $K_{XX}^{-1}$ and $K_{YY}^{-1}$ are invertible),
$$\begin{aligned}e_\phi^{n+1} &= -K_{XX}^{-1} K_{XY} e_\psi^n,\\  e_\psi^{n+1} &= -K_{YY}^{-1} K_{YX} e_\phi^{n+1}.\end{aligned}$$
Combining these into a single-variable iteration yields
$$e_\psi^{n+1} = -K_{YY}^{-1} K_{YX} \left( -K_{XX}^{-1} K_{XY} e_\psi^n \right).$$
Let $M := K_{YY}^{-1} K_{YX} K_{XX}^{-1} K_{XY}$; then $e_\psi^{n+1} = M e_\psi^n$. The convergence condition is $\rho(M) < 1$.

We now examine the mechanism by which the joint distribution is decoupled from marginal observations, and the legitimacy of the above update rule. By analogy with the Rank–Nullity Theorem, we expect that when the lifting operator $\begin{pmatrix}P_{\mathcal X}^*&P_{\mathcal Y}^*\end{pmatrix}$ acts, information is not compressed, which is a property depends on the coercivity condition. Fortunately, within this framework, lossless information transport, the well-definedness of the update rule, and the convergence of the iterative process are all equivalent, and the solution is unique. In other words, the error propagation equation is technically justified.

Reconstructing the full joint state from marginal observations requires that $K_{\mathcal X\mathcal X} = P_\mathcal X G P_\mathcal X^*$ be invertible. By Lemma \ref{lem: CoercivityToInvertible}, this follows from the coercivity of $P_\mathcal X^*$. Moreover,
$$
\langle \phi, K_{XX}\phi \rangle = \langle \phi, P_X G P_X^* \phi \rangle = \langle P_X^*\phi, G P_X^*\phi \rangle.
$$
By Remark \ref{remark: CoercivityOfOperator} and the spectral gap condition on $G$, we have
$$
\langle P_\mathcal X^*\phi, G P_\mathcal X^*\phi \rangle \ge \gamma \|P_\mathcal X^*\phi\|^2.
$$
It remains to prove $\|P_\mathcal X^*\phi\| \ge C\|\phi\|$; then
$$
\langle \phi, K_{\mathcal X\mathcal X}\phi \rangle =\langle P_\mathcal X^*\phi, G P_\mathcal X^*\phi \rangle \ge \|P_\mathcal X^*\phi\| \ge C\|\phi\|.
$$
Thus the coercivity of the operator $K_{\mathcal X\mathcal X}$ is equivalent to the coercivity of $P_\mathcal X^*$, which is precisely the coercivity Lemma \ref{lem: CoecivityOfLifitingOperator}.

\begin{Lem} \label{lem: CoercivityToInvertible}
	If $\|P_\mathcal X^*\phi\| \ge C\|\phi\|$, then $K_{\mathcal X\mathcal X}$ is bijective, i.e., $K_{\mathcal X\mathcal X}^{-1}$ exists.
\end{Lem}
\begin{proof}
	Coercivity implies that the null space satisfies $\ker K_{\mathcal X\mathcal X}=\{0\}$. Since when $K_{\mathcal X\mathcal X}\phi=0$, coercivity gives
	$$
	0=\langle \phi, K_{\mathcal X\mathcal X}\phi \rangle\geq \gamma \|P_\mathcal X^*\phi\|^2 \geq \gamma C\|\phi\|.
	$$
	so $\phi = 0$. Since we are in the linear space $(\mathcal E_\mathcal X^0)^*$, this means that the operator $K_{\mathcal X\mathcal X}$ is injective.
	We first prove that $\operatorname{Im}K_{\mathcal X\mathcal X}$ is closed. By the Cauchy–Schwarz,
	$$
	\|\phi\| \|K_{\mathcal X\mathcal X}\phi\| \ge \langle \phi, K_{\mathcal X\mathcal X}\phi \rangle \ge c \|\phi\|^2,
	$$
	i.e., $\|K_{\mathcal X\mathcal X}\phi\| \ge c \|\phi\|$, for all $\phi$. By Remark \ref{remark: ProjOperatorIsConti} and Remark \ref{remark: OperatorK is Self-Ajoint}, it follows analogously that $K_{\mathcal X\mathcal X} = P_\mathcal X \circ G \circ P_X^*$ is bounded (hence continuous). Take any convergent sequence $w_n = K_{\mathcal X\mathcal X}\phi_n$ with $w_n \to w$. From $\|K_{\mathcal X\mathcal X}\phi\| \ge c \|\phi\|$,
	$$
	\|\phi_n - \phi_m\| \le \frac{1}{c} \|K_{\mathcal X\mathcal X}(\phi_n - \phi_m)\| = \frac{1}{c} \|w_n - w_m\| \to 0.
	$$
	Thus $\{\phi_n\}\in (\mathcal E_\mathcal X^0)^*$ is a Cauchy sequence. By Remark \ref{remark: AdjointOperator}, $(\mathcal E_\mathcal X^0)^*$ is complete, so there exists a unique $\phi$ such that $\phi_n \to \phi$. Since $K_{\mathcal X\mathcal X}$ is continuous, we have $K_{\mathcal X\mathcal X}\phi_n \to K_{\mathcal X\mathcal X}\phi$. But $K_{\mathcal X\mathcal X}\phi_n = w_n \to w$; by the uniqueness of limits in Hausdorff spaces, we obtain $w = K_{\mathcal X\mathcal X}\phi\in \operatorname{Im}K_{\mathcal X\mathcal X}$. Therefore $\operatorname{Im}K_{\mathcal X\mathcal X}$ is closed.
	From $\ker K_{\mathcal X\mathcal X} = \{0\}$, we obtain $(\ker K_{\mathcal X\mathcal X})^\perp = \mathcal E_\mathcal X^0$. By the closed range theorem and the self-adjointness of $K_{\mathcal X\mathcal X}$,
	$$
	\operatorname{Im}K_{\mathcal X\mathcal X}=(\ker K_{\mathcal X\mathcal X}^*)^\perp =(\ker K_{\mathcal X\mathcal X})^\perp =\mathcal E_\mathcal X^0.
	$$
	Hence the range is the entire space, so $K_{\mathcal X\mathcal X}$ is surjective. Therefore $K_{\mathcal X\mathcal X}$ is bijective, i.e., $K_{\mathcal X\mathcal X}^{-1}$ exists.
\end{proof}

From Remark \ref{remark: AdjointOperatorInformationTrans}, when $\phi(x)$ is lifted to $\phi(x,y)$, there is no oscillation in the $y$-direction, so the marginalization integral does not cause any cancellation or loss of edge information. That is, the direction $\phi(x)$ is rigidly embedded in the product space $\mathcal X \times \mathcal Y$, and will not be smoothed out by the marginalization projection. Lemma \ref{lem: CoecivityOfLifitingOperator} further shows that the operator $P_\mathcal X^*$ preserves a fixed proportion of information in every direction. Consequently, on the subspace, the edge information contained in $\phi(x)$ is fully preserved in the product space $\mathcal E_{\mathcal X}^0 \times \mathcal E_{\mathcal Y}^0$, thereby enabling the decoupling of the full joint state. Thus the well-definedness of the update rule and the lossless transport of information are in fact two sides of the same coin.

\begin{Lem}[Coercivity]\label{lem: CoecivityOfLifitingOperator}
	The lifting operators $P_{\mathcal X}^*: \mathcal E_{\mathcal X}^0 \to \mathcal E_{\mathcal X}^0 \times \mathcal E_{\mathcal Y}^0$ and $P_{\mathcal Y}^*: \mathcal E_{\mathcal X}^0 \to \mathcal E_{\mathcal X}^0 \times \mathcal E_{\mathcal Y}^0$ satisfy respectively:
	$$
	\begin{aligned}\|P_{\mathcal X}^*\phi\|_{\mathcal E_0} = \sqrt{\nu(\mathcal Y)} \, \|\phi\|_{\mathcal E_{\mathcal X}^0}, \quad \forall \phi \in \mathcal E_{\mathcal X}^0. \\ \|P_{\mathcal Y}^*\psi\|_{\mathcal E_0} = \sqrt{\mu(\mathcal X)} \, \|\psi\|_{\mathcal E_{\mathcal Y}^0}, \quad \forall \psi \in \mathcal E_{\mathcal Y}^0. \end{aligned}
	$$
	Hence there exist constants $C = \sqrt{\nu(\mathcal Y)} > 0$ and $C' = \sqrt{\mu(\mathcal X)} > 0$ such that
	$$
	\|P_{\mathcal X}^*\phi\|_{\mathcal E_0} \ge C \|\phi\|_{\mathcal E_{\mathcal X}^0}, \quad \|P_{\mathcal Y}^*\psi\|_{\mathcal E_0} \ge C' \|\psi\|_{\mathcal E_{\mathcal Y}^0}.
	$$
\end{Lem}
\begin{proof}
	The norm of $P_{\mathcal X}^* \phi = \phi(x,y)$ is given by $\|P_{\mathcal X}^*\phi\|_{\mathcal E_0}^2= \|P_{\mathcal X}^*\phi\|_{L^2(\mathcal X \times \mathcal Y)}^2 + \|L^{1/2}(P_{\mathcal X}^*\phi)\|_{L^2(\mathcal X \times \mathcal Y)}^2$.
	For the first term $\|P_{\mathcal X}^*\phi\|_{L^2(\mathcal X \times \mathcal Y)}^2$:
	$$
	\|P_{\mathcal X}^*\phi\|_{L^2(\mathcal X \times \mathcal Y)}^2
	= \int_{\mathcal X} \int_{\mathcal Y} |\phi(x)|^2 \, d\nu(y) \, d\mu(x)
	= \nu(\mathcal Y) \int_{\mathcal X} |\phi(x)|^2 \, d\mu(x)
	= \nu(\mathcal Y) \|\phi\|_{L^2(\mathcal X)}^2.
	$$
	For the second term, since $L^{1/2}(P_{\mathcal X}^*\phi)(x,y) = (L_{\mathcal X}^{1/2}\phi)(x)$, we have
	$$
	\|L^{1/2}(P_{\mathcal X}^*\phi)\|_{L^2(\mathcal X \times \mathcal Y)}^2
	= \int_{\mathcal X} \int_{\mathcal Y} |(L_{\mathcal X}^{1/2}\phi)(x)|^2 \, d\nu(y) \, d\mu(x)
	= \nu(\mathcal Y) \|L_{\mathcal X}^{1/2}\phi\|_{L^2(\mathcal X)}^2.
	$$
	Combining these yields
	$$
	\|P_{\mathcal X}^*\phi\|_{\mathcal E_0}^2
	= \nu(\mathcal Y) \left( \|\phi\|_{L^2(\mathcal X)}^2 + \|L_{\mathcal X}^{1/2}\phi\|_{L^2(\mathcal X)}^2 \right)
	= \nu(\mathcal Y) \|\phi\|_{\mathcal E_{\mathcal X}^0}^2.
	$$
	Therefore,
	$$
	\|P_{\mathcal X}^*\phi\|_{\mathcal E_0} = \sqrt{\nu(\mathcal Y)} \, \|\phi\|_{\mathcal E_{\mathcal X}^0}.
	$$
	Taking $C = \sqrt{\nu(\mathcal Y)} > 0$ gives $\|P_{\mathcal X}^*\phi\|_{\mathcal E_0} \ge C \|\phi\|_{\mathcal E_{\mathcal X}^0}$.
	For $P_{\mathcal Y}^*$, the same argument yields
	$$
	\|P_{\mathcal Y}^*\psi\|_{\mathcal E_0} = \sqrt{\mu(\mathcal X)} \, \|\psi\|_{\mathcal E_{\mathcal Y}^0}.
	$$
	Taking $C' = \sqrt{\mu(\mathcal X)} > 0$ gives $\|P_{\mathcal X}^*\phi\|_{\mathcal E_0} \ge C \|\phi\|_{\mathcal E_{\mathcal X}^0}$.
\end{proof}

We now examine the conditions required for iterative convergence. By Remark \ref{remark: OperatorK is Self-Ajoint}, $K$ is a positive definite self-adjoint operator, so its block diagonal entries $K_{\mathcal X\mathcal X}$ and $K_{\mathcal Y\mathcal Y}$ are also positive definite and self-adjoint. For a positive definite block matrix $K = \begin{pmatrix} K_{\mathcal X\mathcal X} & K_{\mathcal X\mathcal Y} \\ K_{\mathcal Y\mathcal X} & K_{\mathcal Y\mathcal Y} \end{pmatrix} > 0$, its Schur complement $S := K_{\mathcal Y\mathcal Y} - K_{\mathcal Y\mathcal X} K_{\mathcal X\mathcal X}^{-1} K_{\mathcal X\mathcal Y}$ is also positive definite, i.e., $S > 0$, which is equivalent to
$$K_{\mathcal Y\mathcal X} K_{\mathcal X\mathcal X}^{-1} K_{\mathcal X\mathcal Y} < K_{\mathcal Y\mathcal Y}.$$
We analyze the estimate of the spectral radius $\rho(M)$ from the Schur complement. Proposition \ref{prop: legitimacy ensures convergence} states that the well-definedness of the update rule guarantees the convergence of the iterative process.

\begin{Prop}\label{prop: legitimacy ensures convergence} 
	When $K_{\mathcal X\mathcal X}^{-1}$ and $K_{\mathcal Y\mathcal Y}^{-1}$ are invertible, we have $\rho(M) < 1$.
\end{Prop}
\begin{proof}
	Lemma \ref{lem: CoercivityToInvertible} and Lemma \ref{lem: CoecivityOfLifitingOperator} guarantee that the following operations are legitimate. Multiplying the Schur complement on the left by $K_{\mathcal Y\mathcal Y}^{-1/2}$ and on the right by $K_{\mathcal Y\mathcal Y}^{-1/2}$, we obtain $K_{\mathcal Y\mathcal Y}^{-1/2} K_{\mathcal Y\mathcal X} K_{\mathcal X\mathcal X}^{-1} K_{\mathcal X\mathcal Y} K_{\mathcal Y\mathcal Y}^{-1/2} < I$. Hence,
	$$
	\rho\left( K_{\mathcal Y\mathcal Y}^{-1/2} K_{\mathcal Y\mathcal X} K_{\mathcal X\mathcal X}^{-1} K_{\mathcal X\mathcal Y} K_{\mathcal Y\mathcal Y}^{-1/2} \right) < 1.
	$$
	Since similarity transformations preserve eigenvalues, we have $\rho(M) < 1$.
\end{proof}

\begin{remark} 
	We elaborate on the similarity structure in Proposition \ref{prop: legitimacy ensures convergence}. Let $S = K_{\mathcal Y\mathcal Y}^{1/2}$, $\widetilde M = S M S^{-1} =  M$, $C = K_{\mathcal X\mathcal X}^{-1/2} K_{\mathcal X\mathcal Y} K_{\mathcal Y\mathcal Y}^{-1/2}$, and $C^* = K_{\mathcal Y\mathcal Y}^{-1/2} K_{\mathcal Y\mathcal X} K_{\mathcal X\mathcal X}^{-1/2}$; then $C^* C=\widetilde M$. Consequently, $\rho(M) = \rho(\widetilde M) = \|C\|^2=\rho(C^* C) < 1$.
\end{remark}

\bibliographystyle{apalike}
\bibliography{reference}

\begin{thebibliography}{}

\bibitem[Blondel et~al., 2018]{blondel2018smooth}
Blondel, M., Seguy, V., and Rolet, A. (2018).
\newblock Smooth and sparse optimal transport.
\newblock In {\em International conference on artificial intelligence and
  statistics}, pages 880--889. PMLR.

\bibitem[Chizat et~al., 2018]{chizat2018unbalanced}
Chizat, L., Peyr{\'e}, G., Schmitzer, B., and Vialard, F.-X. (2018).
\newblock Unbalanced optimal transport: Dynamic and kantorovich formulations.
\newblock {\em Journal of Functional Analysis}, 274(11):3090--3123.

\bibitem[Folland, 2016]{folland2016course}
Folland, G.~B. (2016).
\newblock {\em A course in abstract harmonic analysis}.
\newblock CRC press.

\bibitem[Junius and Oosterhaven, 2003]{junius2003solution}
Junius, T. and Oosterhaven, J. (2003).
\newblock The solution of updating or regionalizing a matrix with both positive
  and negative entries.
\newblock {\em Economic systems research}, 15(1):87--96.

\bibitem[Mainini, 2012]{mainini2012description}
Mainini, E. (2012).
\newblock A description of transport cost for signed measures.
\newblock {\em Journal of Mathematical Sciences}, 181(6):837--855.

\bibitem[Piccoli and Rossi, 2014]{piccoli2014generalized}
Piccoli, B. and Rossi, F. (2014).
\newblock Generalized wasserstein distance and its application to transport
  equations with source.
\newblock {\em Archive for Rational Mechanics and Analysis}, 211(1):335--358.

\bibitem[Santambrogio, 2015]{santambrogio2015optimal}
Santambrogio, F. (2015).
\newblock Optimal transport for applied mathematicians: Calculus of variations,
  pdes, and modeling, volume 87 of progress in nonlinear differential equations
  and their applications.

\bibitem[Villani et~al., 2009]{villani2009optimal}
Villani, C. et~al. (2009).
\newblock {\em Optimal transport: old and new}, volume 338.
\newblock Springer.

\bibitem[Wang et~al., 2025]{wang2025biggest}
Wang, T., Yang, H., Weng, J., and Guo, L. (2025).
\newblock The biggest bank may not be the most interconnected: A refined
  entropy-based approach for indicating the direction of interbank flow.
\newblock {\em International Review of Economics \& Finance}, page 104217.

\end{thebibliography}

\appendix

\section{Omitted proof of Lemmas}
\subsection{Proof of Lemma \ref{lem: complexEntropy}}
\begin{proof} We have:
	$$
	\begin{aligned}H(\rho)&=-\int_{\Pi(\mu,\nu)} \rho\log\rho \,d\pi  \\   &= -\int_{\Pi^+(\mu,\nu)} \rho\log\rho \,d\pi  -\int_{\Pi^-(\mu,\nu)} \rho\left(\log(-\rho) +\pi i \right)\,d\pi   \\&= -\int_{\Pi(\mu,\nu)} \rho\log|\rho| \,d\pi -\pi i\int_{\Pi^-(\mu,\nu)} \rho\,d\pi \\ &= -\int_{\Pi(\mu,\nu)} |\rho|\log|\rho| \,d\pi +\int_{\Pi(\mu,\nu)} (|\rho|-\rho)\log|\rho| \,d\pi -\pi i\int_{\Pi^-(\mu,\nu)} \rho\,d\pi \\&=  -\int_{\Pi(\mu,\nu)} |\rho|\log|\rho| \,d\pi -2\int_{\Pi^-(\mu,\nu)} \rho\log|\rho| \,d\pi -\pi i\int_{\Pi^-(\mu,\nu)} \rho\,d\pi,\quad\rho\in\mathbb C\setminus\{0\} \end{aligned}
	$$
\end{proof}

\subsection{Proof of Lemma \ref{lem: UnitedGeneratorFourierTrans}}
\begin{proof}
	By Remark \ref{remark: FourierCoefficient}, we have $L_{\mathcal X}\chi_\xi=m_{\mathcal X}(\xi)\chi_\xi$ and $\widehat{L_{\mathcal X}f}(\xi)=m_{\mathcal X}(\xi)\hat f(\xi)$. We first compute $\int_{\mathcal X}L_{\mathcal X}\pi_y(x)\overline{\chi_\xi(x)}dx=\widehat{L_{\mathcal X}\pi_y}(\xi)$:
	$$
	\widehat{L_{\mathcal X}\pi_y}(\xi)=m_{\mathcal X}(\xi)\widehat{\pi_y}(\xi)  \Longleftrightarrow \int_{\mathcal X}
	L_{\mathcal X}\pi_y(x)\overline{\chi_\xi(x)}dx=m_{\mathcal X}(\xi)\widehat{\pi_y}(\xi).
	$$
	Expanding $\widehat{(L_{\mathcal X}\otimes I)\pi}(\xi,\eta)$ then yields
	$$
	\begin{aligned}
		\widehat{(L_{\mathcal X}\otimes I)\pi}(\xi,\eta)  =
		&\int\int L_{\mathcal X}\pi(x,y)\overline{\chi_\xi(x)}\overline{\chi_\eta(y)}dxdy \\ =
		&\int_\mathcal{Y}\left[\int_\mathcal{X} L_{\mathcal X}\pi(x,y)\overline{\chi_\xi(x)}dx\right]\overline{\chi_\eta(y)}dy  \\ =&\,m_{\mathcal X}(\xi)\int_\mathcal{Y}\widehat{\pi}^{\,\mathcal X}(\xi,y)\overline{\chi_\eta(y)}dy\\  =&\,m_{\mathcal X}(\xi)\int_\mathcal{Y}\int_\mathcal{X}\pi(x,y)\overline{\chi_\xi(x)}\overline{\chi_\eta(y)}dxdy \\=&\,m_{\mathcal X}(\xi)\int_{X\times Y}\pi(x,y)\overline{\chi_\xi(x)\chi_\eta(y)}dxdy  \\ =&\,m_{\mathcal{X}}(\xi)\,\hat{\pi}(\xi,\eta).
	\end{aligned}
	$$
	where $\widehat{\pi}^{\,\mathcal X}(\xi,y)=\int_{\mathcal X}\pi(x,y)\overline{\chi_\xi(x)}dx$. Similarly, for $I\otimes L_{\mathcal Y}$, we obtain
	$$
	\widehat{(I\otimes L_{\mathcal Y})\pi}(\xi,\eta)=m_{\mathcal Y}(\eta)\hat\pi(\xi,\eta).
	$$
	Taking the Fourier transform of $L_{\mathcal{X}\times\mathcal{Y}} = L_{\mathcal{X}} \otimes I + I \otimes L_{\mathcal{Y}}$ gives
	$$
	\begin{aligned}\widehat{L_{\mathcal X\times\mathcal Y}\pi}&=\widehat{(L_{\mathcal X}\otimes I)\pi}+\widehat{(I\otimes L_{\mathcal Y})\pi}\\&=m_{\mathcal X}(\xi)\hat\pi+m_{\mathcal Y}(\eta)\hat\pi\\&=(m_{\mathcal X}(\xi)+m_{\mathcal Y}(\eta))\hat\pi(\xi,\eta). \end{aligned}
	$$
\end{proof}

\subsection{Proof of Lemma \ref{lem: RelationshipsBetweenEnergySpaces}}
\begin{proof}
	We first prove $D(L_{\mathcal X\times\mathcal Y}) \subset D(L_{\mathcal X\times\mathcal Y}^{1/2})$. Take any $\pi \in D(L_{\mathcal X\times\mathcal Y})$. By definition, we have
	$$
	\int m(\xi,\eta)^2 \,|\hat{\pi}(\xi,\eta)|^2 \,d\hat{\mu}d\hat{\nu} < \infty.
	$$
	We need to show that
	$$
	\int m(\xi,\eta) \,|\hat{\pi}(\xi,\eta)|^2 \,d\hat{\mu}d\hat{\nu} < \infty.
	$$
	Partition the frequency space into two disjoint regions:
	\begin{enumerate}
		\item[(1)] $0 \le m(\xi,\eta) \le 1$. In this region, $m \le 1$, so $m\,|\hat{\pi}|^2 \le |\hat{\pi}|^2$. Since $\pi \in L^2$, Plancherel's theorem gives $\int |\hat{\pi}|^2 < \infty$, hence the integral over this region is finite.
		\item[(2)] $m(\xi,\eta) > 1$. In this region, $m \le m^2$, so $m\,|\hat{\pi}|^2 \le m^2\,|\hat{\pi}|^2$, and the integral of $m^2|\hat{\pi}|^2$ over this region is guaranteed to be finite by $\pi \in D(L_{\mathcal X\times\mathcal Y})$.
	\end{enumerate}
	Combining the two regions yields
	$$
	\int m(\xi,\eta) \,|\hat{\pi}(\xi,\eta)|^2 \,d\hat{\mu}d\hat{\nu} < \infty.
	$$
	Hence $\pi \in D(L_{\mathcal X\times\mathcal Y}^{1/2})$. Therefore the inclusion holds.
	The reverse inclusion $D(L_{\mathcal X\times\mathcal Y}^{1/2}) \subset L^2(\mathcal X\times\mathcal Y)$ is trivial, since by definition $D(L_{\mathcal X\times\mathcal Y}^{1/2})$ is itself a subset of $L^2$. Thus we have
	$$
	D(L_{\mathcal X\times\mathcal Y})\subset\mathcal E=D(L_{\mathcal X\times\mathcal Y}^{1/2})\subset L^2(\mathcal X\times\mathcal Y).
	$$
\end{proof}

\subsection{Proof of Lemma \ref{lem: ConvexityOfDirletEnergy}}
\begin{proof}
	$\forall \pi_1, \pi_2 \in \mathcal M_0 \cap \mathcal E$, $\pi_1 \neq \pi_2$, and $t \in (0,1)$, by the parallelogram identity for the norm,
	$$
	\begin{aligned}
		\mathcal H(t\pi_1 + (1-t)\pi_2)
		&= t\mathcal H(\pi_1) + (1-t)\mathcal H(\pi_2)  - \frac{t(1-t)}{2} \|L^{1/2}(\pi_1 - \pi_2)\|_{L^2}^2  \\&=  t\mathcal H(\pi_1) + (1-t)\mathcal H(\pi_2)  - \frac{t(1-t)}{2} \int_{\widehat{\mathcal X}\times\widehat{\mathcal Y}} m'(\xi,\eta)|\hat{\pi}_1 - \hat{\pi}_2|^2d\hat{\mu}d\hat{\nu}.
	\end{aligned}
	$$
	Since $\pi_1 - \pi_2 \in \mathcal M_0 \setminus \{0\}$ and the spectral gap condition on $\mathcal M_0$ (Definition \ref{def: generator}) implies $m'(\xi,\eta) > 0$, we have
	$$
	\mathcal H(\pi_t) < t\mathcal H(\pi_1) + (1-t)\mathcal H(\pi_2).
	$$
\end{proof}

\section{Omitted proof of Propositions}
\subsection{Proof of Proposition \ref{prop: unboundedness}}\label{sect: proof of unboundedness}
\begin{proof}
	This is equivalent to showing that the feasible set $\Pi(\mu,\nu)$ is unbounded. If $q$ satisfies the marginal constraints, then so does $t q$ for any $t$, since the constraints are homogeneous. We now construct the unboundedness below of the feasible set.
	Suppose there exists a nonzero zero-marginal measure $q$ such that $\langle c, q \rangle = \int c\,dq \neq 0$. Then, along the direction $q$, the objective function takes the value
	$$
	\int c\,d(\pi_0 + t q) = \langle c, \pi_0 \rangle + t \langle c, q \rangle.
	$$
	If $\langle c, q \rangle > 0$, letting $t \to -\infty$ makes the objective function $\int c\,d(\pi_0 + t q) \to -\infty$.
	If $\langle c, q \rangle < 0$, letting $t \to +\infty$ likewise makes the objective function unbounded below.
	Therefore, as long as there exists a nonzero zero-marginal measure $q$ with $\langle c, q \rangle \neq 0$, the feasible set of the original problem is unbounded, and the objective function can be made arbitrarily negative along the zero-marginal direction. In other words, $\mathcal G = \varnothing$.
	When $\langle c, q \rangle = 0$, the optimization problem reduces to a constant term. Along the direction $q$, the original problem admits infinitely many solutions, so $|\mathcal G| = +\infty$.
\end{proof}

\subsection{Proof of Proposition \ref{prop: RichardsonIteration}}
\begin{proof}
	For any two points $z_1, z_2$:
	$$T(z_1) - T(z_2) = (I - \alpha K)(z_1 - z_2).$$
	Therefore,
	$$\|T(z_1) - T(z_2)\| \le \|I - \alpha K\| \cdot \|z_1 - z_2\|.$$
	The Banach fixed-point theorem requires $\|I - \alpha K\| < 1$. We now construct an explicit $\alpha$ that satisfies this condition.
	It is straightforward to verify that $L^{-1}$ is also self-adjoint. Let $\lambda_1, \lambda_2, \dots$ be the spectrum of $L_{\mathcal X\times\mathcal Y}$. Then the spectrum of $G$ is given by $1/\lambda_1, 1/\lambda_2, \dots$. Hence $\|G\| = 1/\lambda_1$. Moreover,
	$$
	K = A G A^*= (A^*)^*G^*A^*  = (A G A^*)^*= K^*,
	$$
	and for all $z \in (\mathcal E_{\mathcal X}^0)^* \times (\mathcal E_{\mathcal Y}^0)^*$,
	$$
	\langle z,Kz\rangle=\langle z,AGA^*z\rangle = \langle A^*z,G A^*z\rangle > 0.
	$$
	Thus $K$ is also a nonnegative self-adjoint operator, with $\|K\| \le \|A\|^2 \|G\| = \frac{\|A\|^2}{\lambda_1}$. 	Therefore its spectrum $\sigma(K)$ is real and bounded: 
	$$
	\sigma(K) \subset [\lambda_{\min}, \lambda_{\max}], \quad 0 < \lambda_{\min} \le \lambda_{\max} < \infty.
	$$
	where $\lambda_{\max} \le \|K\| \le \frac{\|A\|^2}{\lambda_1}$.
	It is easily verified that $I - \alpha K$ is also self-adjoint, so its norm equals the maximum absolute value of its spectrum: $\|I - \alpha K\| = \sup\limits_{\mu \in \sigma(I - \alpha K)} |\mu|$. Since $\sigma(I - \alpha K) = \{1 - \alpha \lambda : \lambda \in \sigma(K)\}$, the norm can be rewritten as
	$$
	\|I-\alpha K\|=\max_{\lambda\in\sigma(K)}|1-\alpha\lambda|.
	$$
	Define $\rho(\alpha) := \|I - \alpha K\| = \max_{\lambda \in [\lambda_{\min}, \lambda_{\max}]} |1 - \alpha \lambda|$. It is straightforward to verify that $\rho(\alpha) < 1$ holds when $0 < \alpha < \frac{2}{\lambda_{\max}}$. The function $|1 - \alpha \lambda|$ is $V$-shaped in $\lambda$ (it is monotone in $\lambda$ on either side of $\lambda = 1/\alpha$), so its maximum is attained at the endpoints $\lambda_{\min}$ and $\lambda_{\max}$. Hence,
	$$
	\max_{\lambda \in [\lambda_{\min}, \lambda_{\max}]} |1 - \alpha \lambda|=\max\{ |1 - \alpha \lambda_{\min}|, |1 - \alpha \lambda_{\max}| \}.
	$$
	We seek the optimal step size $\alpha^* = \arg\min_\alpha \rho(\alpha)$, i.e., the value that minimizes the maximum error factor $\rho(\alpha)$:
	$$
	\min_{\alpha > 0} \max\{ |1 - \alpha \lambda_{\min}|, |1 - \alpha \lambda_{\max}| \}.
	$$
	It is easy to verify that at optimality, the errors at the two endpoints are equal:
	$$
	|1 - \alpha \lambda_{\min}| = |1 - \alpha \lambda_{\max}|.
	$$
	Since $1 - \alpha \lambda$ decreases with $\lambda$, we have $1 - \alpha \lambda_{\min} > 1 - \alpha \lambda_{\max}$. Optimality requires the value to cross zero axis, i.e., $1 - \alpha \lambda_{\min} > 0$ and $1 - \alpha \lambda_{\max} < 0$, Therefore 
	$$
	1 - \alpha \lambda_{\min} = -(1 - \alpha \lambda_{\max}).
	$$
	Solving yields the optimal step size:
	$$
	\alpha^* = \frac{2}{\lambda_{\max} + \lambda_{\min}}
	$$
	Taking the left-endpoint value gives the contraction constant:
	$$
	\begin{aligned}\rho = 1 - \alpha^* \lambda_{\min}  &= 1 - \frac{2\lambda_{\min}}{\lambda_{\max} + \lambda_{\min}} \\ &= \frac{\lambda_{\max} - \lambda_{\min}}{\lambda_{\max} + \lambda_{\min}} \end{aligned}.
	$$
	Since $\lambda_{\max} \ge \lambda_{\min} > 0$, we have $\rho < 1$. This also implies that the contraction condition $\alpha^*\in \left(0,\frac{2}{\lambda_{\max}}\right)$ is automatically satisfied. Therefore, the decoupling equation $z^{n+1} = (I - \alpha^* K)z^n + \alpha^* d$ converges to the unique $z^*$.
\end{proof}

\end{document}